\documentclass[envcountsame]{llncs}
\usepackage{comment}
\usepackage[utf8]{inputenc}
\usepackage{todonotes}
\usepackage{amssymb}
\usepackage{mathrsfs}
\usepackage{graphicx}
\usepackage{amsbsy}
\usepackage{stmaryrd}
\usepackage{tikz}
\usepackage{amsmath}
\usetikzlibrary{calc}
\usepgflibrary{arrows}
\usetikzlibrary{fit}
\usetikzlibrary{patterns}
\usetikzlibrary{decorations.pathreplacing}
\usepackage{xifthen}
\usepackage[
  pdftex,
  bookmarks=true,
  bookmarksnumbered=true,
  hypertexnames=false,
  breaklinks=true,
  linkbordercolor={0 0 1}
  ]{hyperref}
\usepackage{multirow}
\usepackage{manfnt}

\newcommand{\unitint}{I}
\renewcommand{\vec}[1]{\mathbf{#1}}

\newcommand{\restricted}{\upharpoonright}
\newcommand{\N}{\mathbb{N}}
\newcommand{\R}{\mathbb{R}}
\newcommand{\Q}{\mathbb{Q}}

\newcommand{\B}{\mathbb{B}}
\newcommand{\Z}{\mathbb{Z}}

\newcommand{\enc}[1]{\langle #1\rangle}

\newcommand{\defref}[1]{Definition~\ref{#1}}
\newcommand{\figref}[1]{Figure~\ref{#1}}

\newcommand{\lemref}[1]{Lemma~\ref{#1}}
\newcommand{\thref}[1]{Theorem~\ref{#1}}

\newcommand{\dom}{\operatorname{dom}}
\newcommand{\bigO}[1]{\mathcal{O}\left(#1\right)}

\definecolor{darkgreen}{rgb}{0.1,0.6,0.1}
\definecolor{darkyellow}{rgb}{0.6,0.6,0}

\title{On The Complexity of Bounded Time and Precision Reachability for Piecewise Affine Systems\thanks{This work was partially supported by DGA Project CALCULS.}}
\author{Hugo Bazille\inst{3} \and Olivier Bournez\inst{1} \and Walid
  Gomaa\inst{2,4} \and Amaury Pouly\inst{1} }
\institute{École Polytechnique, LIX, 91128 Palaiseau Cedex, France
    \and Egypt Japan University of Science and Technology, CSE, Alexandria, Egypt 
    \and ENS Cachan/Bretagne et Université Rennes 1, France
\and Faculty of Engineering, Alexandria University, Alexandria, Egypt}

\begin{document}

\maketitle

\begin{abstract}
Reachability for piecewise affine systems is known to be undecidable,
starting from dimension $2$. 
In this paper we investigate the exact complexity of several decidable
variants of reachability and control questions
for piecewise affine systems. 
We show in particular that the region-to-region bounded time versions
leads to NP-complete or co-NP-complete problems, starting from
dimension $2$.
We also prove that a bounded precision version leads to
$PSPACE$-complete problems. 


 \end{abstract}

\section{Introduction}

A (discrete time) dynamical system $\mathcal{H}$ is given by some space $X$ and a function $f: X \to X$.
A trajectory of the system
starting from $x_0$ is a sequence $x_0,x_1,x_2,\dots$ etc.,  with
$x_{i+1}=f(x_i)=f^{[i+1]}(x_0)$ where $f^{[i]}$ stands for $i^{th}$ iterate
of $f$. A crucial problem in such systems is the \textit{reachability question}:
given a system $\mathcal{H}$ and $R_0,R \subseteq X$, determine if there is a trajectory starting
from a point of $R_0$ that falls in $R$.
Reachabilty is known to be \textit{undecidable} for very simple functions $f$.
Indeed, it is well-known that various types of dynamical systems, such as hybrid systems, piecewise affine systems, or saturated linear
systems, can simulate Turing machines, see e.g., \cite{KCG94,HKPV95,Moo91,SS95}.

This question is at the heart of \textit{verification} of systems.
Indeed, a safety property corresponds to the determination if there is a trajectory
starting from some set $R_0$ of possible initial states to the set $R$
of bad states.
The industrial and economical impact of having
efficient computer tools, that are able to guarantee that a given system does satisfy its
specification, have indeed generated very important literature. Particularly, many undecidability and
complexity-theoretic results about the hardness of verification of safety properties have been obtained
in the model checking community. 
However, as far as we know, the exact complexity of \textit{natural
restrictions} of the reachability question for systems as simple as
piecewise affine maps are not known, despite their practical interest.

Indeed, existing results mainly focus on the frontier between decidability 
and undecidability. For example, it is known that reachability is undecidable for piecewise constant derivative systems
of dimension $3$, whereas it is decidable for dimension $2$
\cite{AMP95}.
It is known that piecewise affine maps of dimension $2$
can simulate Turing machines \cite{KCGComputability}, whereas the question 
for dimension $1$ is still open and can be related to other natural problems \cite{AsaSch02,Asarin:2001:DRP,BC13}.
Variations of such problems
over the integers have recently been investigated
\cite{ben2013mortality}.

Some complexity facts  follow immediately from these (un)computability
results: for example,
point-to-point bounded time reachability for piecewise affine maps is
$P$-complete as it corresponds to configuration to configuration
reachability for Turing machines. 

However, there remain many natural variants of reachability questions
for 
whose complexity have not yet been established.

For example, in the context of verification,  proving the safety of a
system can often be reduced to a reachability question of the form point-to-region or
region-to-region reachability. These variants are more general 
questions than point-to-point reachability. Their complexities do
not follow from existing results. 

In this paper we choose to restrict to the case of piecewise affine maps
and we consider 
the following natural variant of the problem.
\medskip

\noindent\textbf{
BOUNDED TIME:}
we want to know if region $R$ is reached in less than some prescribed time $T$,
with $f$ assumed to be continuous.

\bigskip

\noindent \textbf{FIXED PRECISION:}  given some precision
  $\varepsilon=2^{-n}$, we want to know if region $R$ is reached by 
  $\tilde{f}_{\varepsilon}$ such that $\tilde{f}_{\varepsilon}(x) =
  \lfloor \frac{f(x)}{\varepsilon}\rfloor\varepsilon$. This
  corresponds to truncating $f$ at precision $2^{-n}$ on each coordinate, or equivalently
  assuming that computations happen at some precision not more that
  $2^{-n}$, for some given $n$. In other words, one wants to know if
  $R$ is reached by the dynamics where precision operation is applied
  at each iteration. 

\begin{remark}
We consider a version where everything is rounded downwards to a
multiple of epsilon. Variants could also be considered, with for
example closest or upper rounding. This would not change the
complexity. 
\end{remark}

\begin{remark}
A variant could also be \emph{chain reachability}: Deciding the existence of
sequence $x_n$ from initial region to target region such that at any
intermediate step $i$, $||x_{i+1}-f(x_i) || \le \epsilon$.  We do not
know the complexity of variants based on this idea. 
\end{remark}

\begin{remark}
We consider piecewise affine maps over the domain $[0,1]^d$,
The case of integer domains has been studied in \cite{ben2013mortality} and turns out to be
quite different. We also assume $f$ to be continuous. This makes the hardness result more natural. 
\end{remark}

In an orthogonal way, control of systems or constructions of controllers for systems often yield to dual questions. 
Instead of asking if some trajectory reaches region $R$, one wants to know if all trajectories reach $R$.
The questions of stability, mortality, or nilpotence for piecewise affine maps and saturated linear systems
have been established in \cite{BBKPC00}.
Still in this context, the
complexity of the problem when restricting to bounded time or fixed precision is not known.

This paper provides an exact characterization of the \textit{algorithmic complexity}
of those two types of reachability for discrete time dynamical systems.
Let $PAF_d$ denote the set of piecewise-affine \emph{continuous} functions
over $[0,1]^d$.
At the end we get the following picture.

\begin{remark} Notice that we expect in
several statements time to be given in unary, in order to get a
completeness result. We do not know about the complexity of the
variants where time would be given in binary. 
\end{remark}

\bigskip

\noindent\textbf{Problem:} \texttt{REACH-REGION}\\
\textbf{Inputs:} a continuous $PAF_d$ $f$ and two regions  $R_0$ and $R$ in $\dom(f)$\\
\textbf{Output:} $\exists x_0\in R_0,t\in\N, 
f^{[t]}(x_0)\in R$?
\begin{theorem}[\cite{KCGComputability}]
Problem \texttt{REACH-REGION} is undecidable (and is 
recursively enumerable-complete). 
\end{theorem}

\noindent\textbf{Problem:} \texttt{CONTROL-REGION}\\
\textbf{Inputs:} a continuous $PAF_d$ $f$ and two regions  $R_0$ and $R$ in $\dom(f)$\\
\textbf{Output:} $\forall x_0\in R_0,\exists t \in \N, 
f^{[t]}(x_0)\in R$?

\begin{theorem}[\cite{BBKPC00}]
Problem \texttt{CONTROL-REGION} is undecidable (and is co-recursively
enumerable complete) for $d\geqslant2$.
\end{theorem}

\noindent\textbf{Problem:} \texttt{REACH-REGION-TIME}\\
\textbf{Inputs:} a time $T\in\N$ in unary, a continuous $PAF_d$ $f$ and two regions  $R_0$ and $R$ in $\dom(f)$\\
\textbf{Output:} $\exists x_0\in R_0, \exists t\leqslant T, f^{[t]}(x_0)\in R$?

\begin{theorem}[Theorems
  \ref{th:reach_np_hard} and \ref{th:reach_in_np}]\label{th:reach_reg_time_np}
\texttt{REACH-REGION-TIME} is NP-complete for $d\geqslant2$.
\end{theorem}

\noindent\textbf{Problem:} \texttt{CONTROL-REGION-TIME}\\
\textbf{Inputs:} a time $T\in\N$ in unary, a continuous $PAF_d$ $f$ and two regions  $R_0$ and $R$ in $\dom(f)$\\
\textbf{Output:} $\forall x_0\in R_0, \exists t\leqslant T, f^{[t]}(x_0)\in R$?

\begin{theorem}[Theorems
  \ref{th:coreach_conp_hard} and \ref{th:coreach_conp}]
\label{th:ctl_reg_time_conp}
\texttt{CONTROL-REGION-TIME} is coNP-complete for $d\geqslant2$.
\end{theorem}

\noindent\textbf{Problem:} \texttt{REACH-REGION-PRECISION}\\
\textbf{Inputs:} a continuous $PAF_d$ $f$, two regions $R_0$ and $R$ in $\dom(f)$
and $\varepsilon=2^{-n}$, $n$ given in unary\\
\textbf{Output:} $\exists x_0\in R_0,t\in\N, (\tilde{f}_\varepsilon)^{[t]}(x_0)\in R$?\\

\noindent\textbf{Problem:} \texttt{CONTROL-REGION-PRECISION}\\
\textbf{Inputs:} a continuous $PAF_d$ $f$, two regions $R_0$ and $R$ in $\dom(f)$
and $\varepsilon=2^{-n}$, $n$ given in unary\\
\textbf{Output:} $\forall x_0\in R_0,\exists t\in\N, (\tilde{f}_\varepsilon)^{[t]}(x_0)\in R$?

\begin{theorem}[Theorems \ref{th:reach_reg_prec_in_pspace} and \ref{th:reach_reg_prec_in_pspace}]
\label{th:approx_reg_pspace}
\texttt{REACH-REGION-PRECISION} is $PSPACE$-complete for
$d\geqslant2$.
\end{theorem}

\begin{theorem}[Theorems \ref{th:reach_reg_prec_in_pspace} and \ref{th:reach_reg_prec_in_pspace}]
\texttt{CONTROL-REGION-PRECISION} is $PSPACE$-complete for
$d\geqslant2$.
\end{theorem}

All our problems are 
region-to-region reachability questions, and requires new proof techniques.

Indeed, classical tricks to simulate a Turing machine using a piecewise affine
maps  encode a Turing machine configuration by a
point, and 
assume that all the points of the trajectories encode (possibly ultimately) valid Turing
machines configurations. 

 This is not a problem in the context of
point-to-point reachability, but this can not be extended to region-to-region reachability. Indeed,  a (non-trivial) region   
consists mostly in invalid points: almost all points do not correspond to
encoding of Turing machines for all the encodings considered in
\cite{KCGComputability,BBKPC00}. 

In order to establish hardness results, the trajectories of all (valid
and invalid) points must be carefully
controlled. This turns out not to be easily possible using the classical
encodings.

Let us insist on the fact that we restrict our results to continuous dynamics.   In this context,
this is an  additional source of difficulties:  
Dealing with
points and trajectories not corresponding to valid configurations or
evolutions.

A short version of this paper has been presented at the conference
``Reachability Problems 2014'' \cite{RP14}. The current journal version
contains full proofs for all statements, and is also providing new
results: bounded precision variants (problems \texttt{REACH-REGION-PRECISION}
and \texttt{CONTROL-REGION-PRECISION}) were not considered in short
version \cite{RP14}.

\section{Preliminaries}

\subsection{Notations}

The set of non-negative integers is denoted $\N$ and the set of the first $n$
naturals is denoted $\N_n=\{0,1,\ldots,n-1\}$.
For any finite set $\Sigma$, let $\Sigma^*$ denote the set of finite words over $\Sigma$. For any word
$w\in\Sigma^*$, let $|w|$ denote the length of $w$. 
Finally, let $\lambda$ denote the empty word.
If $w$ is a word, let $w_1$ denote its first character, $w_2$ the second one and so on.
For any $i,j \in \N$, let $w_{i\ldots j}$ denote the subword $w_iw_{i+1}\ldots w_j$.
For any $\sigma\in\Sigma$, and $k\in\N$, let  $\sigma^k$ denote the word of length $k$
where all symbols are $\sigma$.
For any function $f$, let $f\restricted E$ denote the restriction of $f$ to $E$ and let $\dom(f)$ denote
the domain of definition of $f$.

\subsection{Piecewise affine functions}

Let $\unitint$ denote the unit interval $[0,1]$. Let $d \in \N$. 
A convex closed polyhedron in the space $\unitint^d$ is the solution set of some linear system of inequalities:
\begin{equation}
\begin{aligned}
\label{eqn: convex closed polyhedron}
A\vec{x} \le \vec{b}
\end{aligned}
\end{equation}
with coefficient matrix $A$ and offset vector $\vec{b}$. Let $PAF_d$
denote the set of piecewise-affine continuous functions over
$\unitint^d$: That is to say, 
 any $f\colon\unitint^d\to\unitint^d$ in $PAF_d$, $f$ satisfies:

\begin{itemize}
\item[\textbullet] $f$ is continuous,
\item[\textbullet] there exists a sequence $(P_i)_{1\le i\le p}$ of convex closed polyhedra
with nonempty interior such that $f_i = f\restricted P_i$ is affine,
$\unitint^d = \bigcup_{i=1}^p P_i$ and $\mathring{P_i} \cap \mathring{P_j} = \emptyset$ 
for $i\not=j$, where $\mathring{P}$ denotes the interior of $P$. 
\end{itemize}

In the following discussion we will always assume that any polyhedron $P$ can be defined by a finite set of linear inequalities,  where 
all the elements of $A$ and $\vec{b}$ in  (\ref{eqn: convex closed polyhedron}) are all rationals.
A polyhedron over which $f$ is affine will also be called a region. 

\subsection{Decision problems}

In this paper, we will show hardness results by reduction from known hard problems.
We give the statement of these latter problems in the following.

\medskip

\noindent\textbf{Problem:} \texttt{SUBSET-SUM}\\
\textbf{Inputs:} a goal $B\in\N$ and integers $A_1,\ldots,A_n\in\N$.\\
\textbf{Output:} $\exists I\subseteq\{1,\ldots,n\}, \sum_{i\in I}A_i=B$?

\begin{theorem}[\cite{GJ79}]\label{th:subset_sum_npc}
\texttt{SUBSET-SUM} is NP-complete.
\end{theorem}

\noindent\textbf{Problem:} \texttt{NOSUBSET-SUM}\\
\textbf{Inputs:} a witness $B\in\N$ and integers $A_1,\ldots,A_n\in\N$.\\
\textbf{Output:} $\forall I\subseteq\{1,\ldots,n\}, \sum_{i\in I}A_i\neq B$?

\begin{theorem}\label{th:nosubset_sum_conpc}
\texttt{NOSUBSET-SUM} is coNP-complete.
\end{theorem}

\begin{proof}Basically the same proof as 
  \thref{th:subset_sum_npc} \cite{GJ79}.
\qed\end{proof}

\noindent\textbf{Problem:} \texttt{LINSPACE-WORD}\\
\textbf{Inputs:} A Linear Bounded  Automaton (i.e. a one-tape TM that does not use any space besides the input) $\mathcal{M}$ and a word $w\in\Sigma^*$.\\
\textbf{Output:} does $\mathcal{M}$ accept $w$?

\begin{theorem}[see e.g. \cite{karp1972reducibility}] 
\texttt{LINSPACE-WORD} is PSPACE-complete.
\end{theorem}

\section{Hardness of Bounded Time Reachability}

In this section, we will show that \texttt{REACH-REGION-TIME} is an NP-hard problem
by reducing \texttt{SUBSET-SUM} to it.

\subsection{Solving \texttt{SUBSET-SUM} by iteration}

We will now show how to solve the \texttt{SUBSET-SUM} problem by iterating a function.
Consider an instance $\mathcal{I}=(B,A_1,\ldots,A_n)$ of \texttt{SUBSET-SUM}. We will
need to introduce some notions before defining our piecewise affine function.
Our first notion is that of configurations, representing partial summation
of the number for a given choice of $I$.

\begin{remark}
Without loss of generality, we will only consider instances where $A_i\leqslant B$,
for all $i$. Indeed, if $A_i>B$, it will never be an element of a
solution to the instance and so we
can simply remove this variable from the problem.
This ensures that $A_i<B+1$ in everything that follows.
\end{remark}

\begin{definition}[Configuration]
A configuration of $\mathcal{I}$ is a tuple $(i,\sigma,\varepsilon_i,\ldots,\varepsilon_n)$
where $i\in\{1,\ldots,n+1\}$, $\sigma\in\{0,\ldots,B+1\}$, $\varepsilon_j\in\{0,1\}$ for
all $i \le j $.
Let $\mathcal{C}_\mathcal{I}$ be the set of all configurations of $\mathcal{I}$.
\end{definition}

The intuitive understanding of a configuration, made formal in the next definition, is the following:
$(i,\sigma,\varepsilon_i,\ldots,\varepsilon_n)$ represents a
situation where after having summed a subset of $\{A_1,\ldots,A_{i-1}\}$, we got a
sum $\sigma$  and $\varepsilon_j$ is $1$ if and only if we are to pick $A_j$ in the future. 

\begin{definition}[Transition function]\label{def:it_T_subset}
The transition function $T_\mathcal{I}:\mathcal{C}_\mathcal{I}\rightarrow\mathcal{C}_\mathcal{I}$, 
is defined as follows:
\[
T_{\mathcal{I}}(i,\sigma,\varepsilon_i,\ldots,\varepsilon_n)=
\begin{cases}
(i,\sigma)&\text{if }i=n+1\\
(i+1,\min\left(B+1,\sigma+\varepsilon_iA_i\right),\varepsilon_{i+1},\ldots,\varepsilon_n)&\text{otherwise}
\end{cases}
\]
\end{definition}

It should be clear, by definition of a subset sum that we have the following
simulation result.

\begin{lemma}\label{lem:subset_sum_it}
For any configuration $c=(i,\sigma,\varepsilon_i,\ldots,\varepsilon_n)$
and $k\in\{0,\ldots,n+1-i\}$,
\[T_\mathcal{I}^{[k]}(c)=(i+k,\min\left(B+1,\sigma+\Sigma_{j=i}^{i+k-1}\varepsilon_j A_j\right),
    \varepsilon_{i+k},\ldots,\varepsilon_n)\]
\end{lemma}
\begin{proof}By induction.\qed\end{proof}

A consequence of this simulation is that we can reformulate
satisfiability in terms of reachability.

\begin{lemma}\label{lem:subset_sum_simul} 
$\mathcal{I}$ is a satisfiable instance (\emph{i.e.,} admits a subset
sum equal to the target value) 
if and only if there exists a configuration $c=(1,0,\varepsilon_1,\ldots,\varepsilon_n)\in\mathcal{C}_\mathcal{I}$
such that $T_\mathcal{I}^{[n]}(c)=(n+1,B)$.
\end{lemma}

\begin{proof}
The ``only if'' direction is the simplest: assume there exists $I\subseteq\{1,\ldots,n\}$
such that $\sum_{i\in I}A_i=B$. Define $\varepsilon_i=1$ if $i\in I$ and $0$ otherwise.
We get that $\sum_{i=1}^n \varepsilon_iA_i=B$. Apply \lemref{lem:subset_sum_it} to get that:
\begin{align*}
T_\mathcal{I}^{[n]}(1,0,\varepsilon_1,\ldots,\varepsilon_n)
    &=(n+1,\min\left(B+1,0+\sum_{i=1}^n\varepsilon_iA_i\right))\\
    &=(n+1,\min(B+1,B))=(n+1,B)
\end{align*}
The ``if'' direction is very similar: assume that there exists
$c=(1,0,\varepsilon_1,\ldots,\varepsilon_n)$
such that $T_\mathcal{I}^{[n]}(c)=(n+1,B)$.   \lemref{lem:subset_sum_it} gives:

\[T_\mathcal{I}^{[n]}(1,0,\varepsilon_1,\ldots,\varepsilon_n)=
(n+1,\min\left(B+1,0+\sum_{i=1}^n\varepsilon_iA_i\right))\]

We can easily conclude that $\sum_{i=1}^n\varepsilon_iA_i=B$
and thus by
defining $I=\{i\thinspace|\thinspace\varepsilon_i=1\}$ we get that $\sum_{i\in I}A_i=B$. Hence, $\mathcal{I}$
is satisfiable.
\qed\end{proof}

\subsection{Solving \texttt{SUBSET-SUM} with a piecewise affine function}\label{sec:subset_sum_affine}

In this section, we explain how to simulate the function $T_\mathcal{I}$
using a piecewise affine function and some encoding of the configurations for
a given $\mathcal{I}=(B,A_1,\ldots,A_n)$. Since the reduction is quite technical,
we start with some intuitions. In order to simulate the function $T_\mathcal{I}$,
we first need to encode configurations with real numbers. Let $c=(i,\sigma,\varepsilon_i,\ldots,\varepsilon_n)$
be a configuration, we encode it using two real numbers in $[0,1]$: the first one encodes $i$
and $\sigma$ and the second one encodes $\varepsilon_i,\ldots,\varepsilon_n$.
A simple approach is to encode them as digits of dyadics numbers, as depicted below:
\[\enc{c}=\begin{pmatrix}
    0.\raisebox{-0.4ex}{\begin{tikzpicture}[scale=1]
        \fill[draw,color=blue!50!white,fill=blue!30!white] (0,0) rectangle (1.5,0.5);
        \fill[draw,color=red!50!white,fill=red!30!white] (1.5,0) rectangle (3,0.5);
        \node[inner sep=0] at (0.75,0.2) {$i$};
        \node[inner sep=0] at (2.25,0.2) {$\sigma$};
        \end{tikzpicture}}\\
    0.\raisebox{-0.4ex}{\begin{tikzpicture}[scale=1]
        \fill[draw,color=darkgreen!50!white,fill=darkgreen!30!white] (0,0) rectangle (3,0.5);
        \foreach \x in {0.5,1,1.5,2.5} \draw[color=darkgreen!50!white] (\x,0) -- (\x,0.5);
        \node[inner sep=0] at (0.25,0.2) {$0$};
        \node[inner sep=0] at (0.75,0.2) {$\ldots$};
        \node[inner sep=0] at (1.25,0.2) {$\varepsilon_i$};
        \node[inner sep=0] at (2,0.2) {$\ldots$};
        \node[inner sep=0] at (2.75,0.2) {$\varepsilon_n$};\end{tikzpicture}}
    \end{pmatrix}
    =\begin{pmatrix}
        i2^{-p}+\sigma2^{-q}\\\varepsilon_i2^{-1}+\varepsilon_{i+1}2^{-2}+\cdots
    \end{pmatrix}.
\]
In the above encoding, we allocate $p$ bits to $i$ and $q-p$ bits to $\sigma$. We
simply need to choose them large enough to accomodate the largest possible value.
The rationale behind this encoding is that it is
easy to implement the transition function $f_\mathcal{I}$ with a linear function.
Note that for technical reasons explained later, we need to encode the second coordinate
in basis $\beta=5$ instead of $2$. We now encode $0$ as $0^\star=1$
and $1$ as $1^\star=4$. Graphically, the action of $f$ is very simple:
\[\text{\textbf{if} }\varepsilon_i=0^\star\text{ \textbf{then} } f_\mathcal{I}\begin{pmatrix}
    0.\raisebox{-0.4ex}{\begin{tikzpicture}[scale=1]
        \fill[draw,color=blue!50!white,fill=blue!30!white] (0,0) rectangle (1.5,0.5);
        \fill[draw,color=red!50!white,fill=red!30!white] (1.5,0) rectangle (3,0.5);
        \node[inner sep=0] at (0.75,0.2) {$i$};
        \node[inner sep=0] at (2.25,0.2) {$\sigma$};
        \end{tikzpicture}}\\
    0.\raisebox{-0.4ex}{\begin{tikzpicture}[scale=1]
        \fill[draw,color=darkgreen!50!white,fill=darkgreen!30!white] (0,0) rectangle (3,0.5);
        \foreach \x in {0.5,1,1.5,2,2.5} \draw[color=darkgreen!50!white] (\x,0) -- (\x,0.5);
        \node[inner sep=0] at (0.25,0.2) {$0$};
        \node[inner sep=0] at (0.75,0.2) {$\ldots$};
        \node[inner sep=0] at (1.25,0.2) {$\varepsilon_i$};
        \node[inner sep=0] at (1.75,0.2) {$\scriptscriptstyle\varepsilon_{i+1}$};
        \node[inner sep=0] at (2.25,0.2) {$\ldots$};
        \node[inner sep=0] at (2.75,0.2) {$\varepsilon_n$};\end{tikzpicture}}
    \end{pmatrix}
    =
    \begin{pmatrix}
    0.\raisebox{-0.4ex}{\begin{tikzpicture}[scale=1]
        \fill[draw,color=blue!50!white,fill=blue!30!white] (0,0) rectangle (1.5,0.5);
        \fill[draw,color=red!50!white,fill=red!30!white] (1.5,0) rectangle (3,0.5);
        \node[inner sep=0] at (0.75,0.2) {$i+1$};
        \node[inner sep=0] at (2.25,0.2) {$\sigma$};
        \end{tikzpicture}}\\
    0.\raisebox{-0.4ex}{\begin{tikzpicture}[scale=1]
        \fill[draw,color=white,fill=white] (2.5,0) rectangle (3,0.5);
        \fill[draw,color=darkgreen!50!white,fill=darkgreen!30!white] (0,0) rectangle (3,0.5);
        \foreach \x in {0.5,1,1.5,2,2.5} \draw[color=darkgreen!50!white] (\x,0) -- (\x,0.5);
        \node[inner sep=0] at (0.25,0.2) {$0$};
        \node[inner sep=0] at (0.75,0.2) {$\ldots$};
        \node[inner sep=0] at (1.25,0.2) {$0$};
        \node[inner sep=0] at (1.75,0.2) {$\scriptscriptstyle\varepsilon_{i+1}$};
        \node[inner sep=0] at (2.25,0.2) {$\ldots$};
        \node[inner sep=0] at (2.75,0.2) {$\varepsilon_n$};\end{tikzpicture}}
    \end{pmatrix},\]
\[\text{\textbf{if} }\varepsilon_i=1^\star\text{ \textbf{then} } f_\mathcal{I}\begin{pmatrix}
    0.\raisebox{-0.4ex}{\begin{tikzpicture}[scale=1]
        \fill[draw,color=blue!50!white,fill=blue!30!white] (0,0) rectangle (1.5,0.5);
        \fill[draw,color=red!50!white,fill=red!30!white] (1.5,0) rectangle (3,0.5);
        \node[inner sep=0] at (0.75,0.2) {$i$};
        \node[inner sep=0] at (2.25,0.2) {$\sigma$};
        \end{tikzpicture}}\\
    0.\raisebox{-0.4ex}{\begin{tikzpicture}[scale=1]
        \fill[draw,color=darkgreen!50!white,fill=darkgreen!30!white] (0,0) rectangle (3,0.5);
        \foreach \x in {0.5,1,1.5,2,2.5} \draw[color=darkgreen!50!white] (\x,0) -- (\x,0.5);
        \node[inner sep=0] at (0.25,0.2) {$0$};
        \node[inner sep=0] at (0.75,0.2) {$\ldots$};
        \node[inner sep=0] at (1.25,0.2) {$\varepsilon_i$};
        \node[inner sep=0] at (1.75,0.2) {$\scriptscriptstyle\varepsilon_{i+1}$};
        \node[inner sep=0] at (2.25,0.2) {$\ldots$};
        \node[inner sep=0] at (2.75,0.2) {$\varepsilon_n$};\end{tikzpicture}}
    \end{pmatrix}
    =
    \begin{pmatrix}
    0.\raisebox{-0.4ex}{\begin{tikzpicture}[scale=1]
        \fill[draw,color=blue!50!white,fill=blue!30!white] (0,0) rectangle (1.5,0.5);
        \fill[draw,color=red!50!white,fill=red!30!white] (1.5,0) rectangle (3,0.5);
        \node[inner sep=0] at (0.75,0.2) {$i+1$};
        \node[inner sep=0] at (2.25,0.2) {$\sigma+A_i$};
        \end{tikzpicture}}\\
    0.\raisebox{-0.4ex}{\begin{tikzpicture}[scale=1]
        \fill[draw,color=white,fill=white] (2.5,0) rectangle (3,0.5);
        \fill[draw,color=darkgreen!50!white,fill=darkgreen!30!white] (0,0) rectangle (3,0.5);
        \foreach \x in {0.5,1,1.5,2,2.5} \draw[color=darkgreen!50!white] (\x,0) -- (\x,0.5);
        \node[inner sep=0] at (0.25,0.2) {$0$};
        \node[inner sep=0] at (0.75,0.2) {$\ldots$};
        \node[inner sep=0] at (1.25,0.2) {$0$};
        \node[inner sep=0] at (1.75,0.2) {$\scriptscriptstyle\varepsilon_{i+1}$};
        \node[inner sep=0] at (2.25,0.2) {$\ldots$};
        \node[inner sep=0] at (2.75,0.2) {$\varepsilon_n$};\end{tikzpicture}}
    \end{pmatrix}.
\]
For technical reasons, we will in fact split the case
$\varepsilon_i=1^\star$ into two, depending on whether $\sigma+A_i$ becomes greater
then $B+1$ or not. This is similar to what we did in \defref{def:it_T_subset} with
$\min(B+1,\sigma+A_i)$:
\[\text{\textbf{if} }\sigma+A_i> B\text{ \textbf{then} } f_\mathcal{I}\begin{pmatrix}
    0.\raisebox{-0.4ex}{\begin{tikzpicture}[scale=1]
        \fill[draw,color=blue!50!white,fill=blue!30!white] (0,0) rectangle (1.5,0.5);
        \fill[draw,color=red!50!white,fill=red!30!white] (1.5,0) rectangle (3,0.5);
        \node[inner sep=0] at (0.75,0.2) {$i$};
        \node[inner sep=0] at (2.25,0.2) {$\sigma$};
        \end{tikzpicture}}\\
    0.\raisebox{-0.4ex}{\begin{tikzpicture}[scale=1]
        \fill[draw,color=darkgreen!50!white,fill=darkgreen!30!white] (0,0) rectangle (3,0.5);
        \foreach \x in {0.5,1,1.5,2,2.5} \draw[color=darkgreen!50!white] (\x,0) -- (\x,0.5);
        \node[inner sep=0] at (0.25,0.2) {$0$};
        \node[inner sep=0] at (0.75,0.2) {$\ldots$};
        \node[inner sep=0] at (1.25,0.2) {$1^\star$};
        \node[inner sep=0] at (1.75,0.2) {$\scriptscriptstyle\varepsilon_{i+1}$};
        \node[inner sep=0] at (2.25,0.2) {$\ldots$};
        \node[inner sep=0] at (2.75,0.2) {$\varepsilon_n$};\end{tikzpicture}}
    \end{pmatrix}
    =
    \begin{pmatrix}
    0.\raisebox{-0.4ex}{\begin{tikzpicture}[scale=1]
        \fill[draw,color=blue!50!white,fill=blue!30!white] (0,0) rectangle (1.5,0.5);
        \fill[draw,color=red!50!white,fill=red!30!white] (1.5,0) rectangle (3,0.5);
        \node[inner sep=0] at (0.75,0.2) {$i+1$};
        \node[inner sep=0] at (2.25,0.2) {$B+1$};
        \end{tikzpicture}}\\
    0.\raisebox{-0.4ex}{\begin{tikzpicture}[scale=1]
        \fill[draw,color=white,fill=white] (2.5,0) rectangle (3,0.5);
        \fill[draw,color=darkgreen!50!white,fill=darkgreen!30!white] (0,0) rectangle (3,0.5);
        \foreach \x in {0.5,1,1.5,2,2.5} \draw[color=darkgreen!50!white] (\x,0) -- (\x,0.5);
        \node[inner sep=0] at (0.25,0.2) {$0$};
        \node[inner sep=0] at (0.75,0.2) {$\ldots$};
        \node[inner sep=0] at (1.25,0.2) {$0$};
        \node[inner sep=0] at (1.75,0.2) {$\scriptscriptstyle\varepsilon_{i+1}$};
        \node[inner sep=0] at (2.25,0.2) {$\ldots$};
        \node[inner sep=0] at (2.75,0.2) {$\varepsilon_n$};\end{tikzpicture}}
    \end{pmatrix}.
\]

We can formulate the ``if $\varepsilon_i=\alpha$'' by using regions. Indeed, the set
of encodings such that $\varepsilon_i=\alpha$ is
\[\left\{\begin{pmatrix}
    0.\raisebox{-0.4ex}{\begin{tikzpicture}[scale=1]
        \fill[draw,color=blue!50!white,fill=blue!30!white] (0,0) rectangle (1.5,0.5);
        \fill[draw,color=red!50!white,fill=red!30!white] (1.5,0) rectangle (3,0.5);
        \node[inner sep=0] at (0.75,0.2) {$i$};
        \node[inner sep=0] at (2.25,0.2) {$\sigma$};
        \end{tikzpicture}}\\
    0.\raisebox{-0.4ex}{\begin{tikzpicture}[scale=1]
        \fill[draw,color=darkgreen!50!white,fill=darkgreen!30!white] (0,0) rectangle (3,0.5);
        \foreach \x in {0.5,1,1.5,2,2.5} \draw[color=darkgreen!50!white] (\x,0) -- (\x,0.5);
        \node[inner sep=0] at (0.25,0.2) {$0$};
        \node[inner sep=0] at (0.75,0.2) {$\ldots$};
        \node[inner sep=0] at (1.25,0.2) {$\alpha$};
        \node[inner sep=0] at (1.75,0.2) {$\scriptscriptstyle\varepsilon_{i+1}$};
        \node[inner sep=0] at (2.25,0.2) {$\ldots$};
        \node[inner sep=0] at (2.75,0.2) {$\varepsilon_n$};\end{tikzpicture}}
    \end{pmatrix}:\sigma\in\N,\varepsilon_j\in\{0,1\}
    \right\}
    \subset\begin{array}{@{}l}
        \big[i2^{-p},i2^{-p}+2^{-p-1}\big]\\
        \times\big[\alpha\beta^{-i},(\alpha+1)\beta^{-i}\big]
    \end{array}.\]
It is crucial to note that the statement of the problems only allows for
polyhedral regions. This is why in the above equation, we had to overapproximate
the region by intervals on each coordinate. This overapproximation is the root
of all difficulties. Indeed, the region now contains many points that do not
correspond to encodings anymore. We now go to the details of the construction.

\begin{definition}[Encoding]
Define $p=\lceil\log_2 (n+2)\rceil$, $\omega=\lceil\log_2 (B+2)\rceil$, $q=p+\omega+1$
and $\beta=5$. Also define $0^\star=1$ and $1^\star=4$.
For any configuration $c=(i,\sigma,\varepsilon_i,\ldots,\varepsilon_n)$, define the encoding of $c$ as follows:

\[\enc{c}=\left(i2^{-p}+\sigma2^{-q},0^\star\beta^{-n-1}+\sum_{j=i}^n\varepsilon_j^\star\beta^{-j}\right)\]
Also define the following regions for any $i\in\{1,\ldots,n+1\}$ and $\alpha\in\{0,\ldots,\beta-1\}$:

\[R_0=[0,2^{-p-1}]\times[0,1]\qquad R_i=[i2^{-p},i2^{-p}+2^{-p-1}]\times[0,\beta^{-i+1}]\quad(i\geqslant1)\]
\[R_{i,\alpha}=\big[i2^{-p},i2^{-p}+2^{-p-1}\big]\times\big[\alpha\beta^{-i},(\alpha+1)\beta^{-i}\big]
\qquad R_i=\cup_{\alpha\in\N_\beta}R_{i,\alpha}\]
\[R_{i,1^\star}^{lin}=\big[i2^{-p},i2^{-p}+(B+1-A_i)2^{-q}\big]\times\big[1^\star\beta^{-i},5\beta^{-i}\big]
\qquad R_{i,1^\star}^{sat}=R_{i,1^\star}\setminus R_{i,1^\star}^{lin}\]
\end{definition}

As noted before, we use basis $\beta=5$ on the second component to get some ``space''
between consecutive encodings. The choice of the value $1$ and $4$ for the encoding of
$0$ and $1$, although not crucial, has been made to simplify the proof as much as possible.

The region $R_0$ is for initialization purposes and is defined differently from
the other $R_i$. The regions $R_i$ correspond to the
different values of $i$ in the configuration (the current number). Each $R_i$ is
further divided into the $R_{i,\alpha}$ corresponding to all the possible values
of the next $\varepsilon$ variable (recall that it is encoded in basis $\beta$).
In the special case of $\varepsilon=1$, we cut the region $R_{i,1^\star}$ into
a linear part and a saturated part. This is needed to emulate the $\min(\sigma+A_i,B+1)$
in \defref{def:it_T_subset}: the linear part corresponds to $\sigma+A_i$
and the saturated part to $B+1$.
\figref{fig:regions} and \figref{fig:region_zoom} give a graphical representation of the regions.

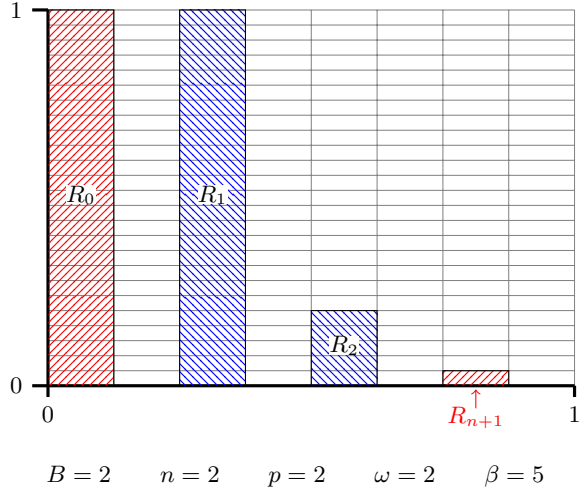
\begin{figure}
\begin{center}
\begin{tikzpicture}[xscale=7,yscale=5]
  \pgfmathsetmacro{\n}{2}
  \pgfmathsetmacro{\maxn}{\n+1}
  \pgfmathsetmacro{\p}{ceil(log2(\n+2))}
  \pgfmathsetmacro{\B}{2}
  \pgfmathsetmacro{\vomega}{ceil(log2(\B+2))}
  \pgfmathsetmacro{\q}{\p+\vomega+1}
  \pgfmathsetmacro{\vbeta}{5}
  \pgfmathsetmacro{\xstep}{pow(2,-\p-1)}
  \pgfmathsetmacro{\ystep}{pow(\vbeta,-\maxn+1)}
  \draw[very thin,color=gray,xstep=\xstep,ystep=\ystep] (0.0,0.0) grid (1,1);
  \draw[very thick] (0,0) -- (1, 0);
  \draw[very thick] (0,0) -- (0, 1);
  \draw[very thick] (0,0) -- (0,-0.03) node[below] {$0$};
  \draw[very thick] (1,0) -- (1,-0.03) node[below] {$1$};
  \draw[very thick] (0,1) -- (-0.03,1) node[left] {$1$};
  \draw[very thick] (0,0) -- (-0.03,0) node[left] {$0$};
  \draw[pattern color=red,pattern=north east lines] (0,0) rectangle (\xstep,1);
  \draw (\xstep/2,.5-0.01) node[anchor=base,fill=white,inner sep=0]{$\displaystyle R_0$};
  \pgfmathsetmacro{\x}{\maxn*\xstep*2};
  \pgfmathsetmacro{\y}{pow(\vbeta,-\maxn+1)};
  \draw[pattern color=red,pattern=north east lines] (\x,0) rectangle (\x+\xstep,\y);
  \draw (\x+\xstep/2,-0.1) node[red,anchor=base,fill=white,inner sep=0]{$\displaystyle R_{n+1}$};
  \draw[red,->] (\x+\xstep/2,-0.06) -- (\x+\xstep/2,-0.01);
  \foreach \i in {1,...,\n}
  {
      \pgfmathsetmacro{\y}{pow(\vbeta,-\i+1)};
      \draw[pattern color=blue,pattern=north west lines] (\i*\xstep*2,0) rectangle (\i*\xstep*2+\xstep,\y);
      \draw (\i*\xstep*2+\xstep/2,\y/2-0.01) node[anchor=base,fill=white,inner sep=0]{$\displaystyle R_{\i}$};
  }
\end{tikzpicture}
\[B=2\qquad n=2\qquad p=2\qquad\omega=2\qquad\beta=5\]
\end{center}
\caption{Graphical representation of the regions\label{fig:regions}}
\end{figure}

\begin{lemma}\label{lem:encoding_prop}
For any configuration $c=(i,\sigma,\varepsilon_i,\ldots,\varepsilon_n)$,
if $i=n+1$ then $\enc{c}\in R_{n+1,0^\star}$, otherwise $\enc{c}\in R_{i,\varepsilon_i^\star}$.
Furthermore if $\varepsilon_i=1$ and $\sigma+A_i\leqslant B+1$, then $\enc{c}\in R_{i,1^\star}^{lin}$,
otherwise $\enc{c}\in R_{i,1^\star}^{sat}$.
\end{lemma}

\begin{proof}
Recall
that $\omega=\lceil\log_2(B+2)\rceil$ so $B+1<2^\omega$, and $q=p+\omega+1$.
Since $\sigma\leqslant B+1$ by definition,
$(n+1)2^{-p}\leqslant\enc{c}_1\leqslant (n+1)2^{-p}+(B+1)2^{-p-1-\omega}\leqslant(n+1)2^{-p}+2^{-p-1}$.
This shows the result for the first component. In the case where $\sigma+A_i\leqslant B+1$
then $\sigma2^{-q}\leqslant (B+1-A_i)2^{-q}$ yielding the result for the second
part of the result for the first component.

If $i=n+1$, then $\enc{c}_2=0^\star\beta^{-p-1}$ and it trivially belongs to 
$[0^\star\beta^{-n-1},(0^\star+1)\beta^{-n-1}]$.
Otherwise, 
\begin{multline*}
\varepsilon_i^\star\beta^{-i}\leqslant \enc{c}_2\leqslant \varepsilon_i^\star\beta^{-i}+
\sum_{j=i+1}^{n+1}1^\star\beta^{-j} 
\leqslant
\varepsilon_i^\star\beta^{-i}+1^\star\beta^{-i-1}\frac{1-\beta^{n-i}}{1-\beta^{-1}} \\
\leqslant
\varepsilon_i^\star\beta^{-i}+4\beta^{-i-1}\frac{\beta}{\beta-1} 
\leqslant \varepsilon_i^\star\beta^{-i}+\beta^{-i} 
\leqslant (\varepsilon_i^\star+1)\beta^{-i}.
\end{multline*} This shows the result when $i<n+1$,
for the second component of the result.
\qed\end{proof}

We can now define a piecewise affine function that will mimic the behavior of
$T^\mathcal{I}$. The region $R_0$ is here to ensure that we start from a ``clean'' value
on the first coordinate.

\begin{definition}[Piecewise affine simulation]\label{def:subset_paf}
\[
f_\mathcal{I}(a,b)=\begin{cases}
(2^{-p},b)&\text{if }(a,b)\in R_0\\
(a,b)&\text{if }(a,b)\in R_{n+1}\\
(a+2^{-p},b-0^\star\beta^{-i})&\text{if }(a,b)\in R_{i,0^\star}\\
(a+2^{-p}+A_i2^{-q},b-1^\star\beta^{-i})&\text{if }(a,b)\in R_{i,1^\star}^{lin}\\
((i+1)2^{-p}+(B+1)2^{-q},b-1^\star\beta^{-i})&\text{if }(a,b)\in R_{i,1^\star}^{sat}\\
\end{cases}
\]
\end{definition}

\begin{lemma}[Simulation is correct]\label{lem:simul_correct}
For any configuration $c\in\mathcal{C}_\mathcal{I}$,
$\enc{T_\mathcal{I}(c)}=f_\mathcal{I}(\enc{c})$.
\end{lemma}

\begin{proof}

Let $c=(i,\sigma,\varepsilon_i,\ldots,\varepsilon_n)$. There are several cases to consider:
if $i=n+1$ then $T_\mathcal{I}(c)=c$, also by \lemref{lem:encoding_prop}, $\enc{c}\in R_{n+1,0^\star}$.
Thus by definition of $f$, $f_\mathcal{I}(\enc{c})=\enc{c}=\enc{T(c)}$
and this shows the result.
If $i<n+1$, we have three more cases to consider: the case where we don't take the value ($\varepsilon_i=0$)
and the two cases where we take it ($\varepsilon_i=1$) with and without saturation.
\begin{itemize}
\item If $\varepsilon_i=0$ then $T_\mathcal{I}(c)=(i+1,\sigma,\varepsilon_{i+1},\ldots,\varepsilon_n)$.
On the other hand, $\enc{c}=(a,b)=(i2^{-p}+\sigma2^{-q},0^\star\beta^{-i}+\sum_{j=i+1}^n\varepsilon_j\beta^{-j}+0^\star\beta^{-n-1})$.
By \lemref{lem:encoding_prop}, $\enc{c}\in R_{i,0^\star}$ so by definition of $f$:
\begin{align*}
f_\mathcal{I}(\enc{c})&=(a+2^{-p},b-0^\star\beta^{-i})\\
    &=((i+1)2^{-p}+\sigma2^{-q},\sum_{j=i+1}^n\varepsilon_j\beta^{-j}+0^\star\beta^{-n-1})\\
    &=\enc{(i+1,\sigma,\varepsilon_{i+1},\ldots,\varepsilon_n)}=\enc{T_\mathcal{I}(c)}
\end{align*}
\item If $\varepsilon_i=1$ and $\sigma+A_i\leqslant B+1$ then
$T_\mathcal{I}(c)=(i+1,\sigma+A_i,\varepsilon_{i+1},\ldots,\varepsilon_n)$.
On the other hand, $\enc{c}=(a,b)=(i2^{-p}+\sigma2^{-q},1^\star\beta^{-i}+\sum_{j=i+1}^n\varepsilon_j\beta^{-j}+0^\star\beta^{-n-1})$.
By \lemref{lem:encoding_prop}, $\enc{c}\in R_{i,1^\star}^{lin}$ so by definition of $f$:
\begin{align*}
f_\mathcal{I}(\enc{c})&=(a+2^{-p}+A_i2^{-q},b-1^\star\beta^{-i})\\
    &=((i+1)2^{-p}+(\sigma+A_i)2^{-q},\sum_{j=i+1}^n\varepsilon_j\beta^{-j}+0^\star\beta^{-n-1})\\
    &=\enc{(i+1,\sigma+A_i,\varepsilon_{i+1},\ldots,\varepsilon_n)}=\enc{T_\mathcal{I}(c)}
\end{align*}
\item If $\varepsilon_i=1$ and $\sigma+A_i>B+1$ then
$T_\mathcal{I}(c)=(i+1,B+1,\varepsilon_{i+1},\ldots,\varepsilon_n)$.
By \lemref{lem:encoding_prop}, $\enc{c}\in R_{i,1^\star}^{sat}$ so by definition of $f$:
\begin{align*}
f_\mathcal{I}(\enc{c})&=((i+1)2^{-p}+(B+1)2^{-q},b-1^\star\beta^{-i})\\
    &=\enc{(i+1,B+1,\varepsilon_{i+1},\ldots,\varepsilon_n)}=\enc{T_\mathcal{I}(c)}
\end{align*}
\end{itemize}
\qed\end{proof}

\subsection{Making the simulation stable}

In the previous section, that we have defined $f_\mathcal{I}$ over a subset of the entire space and it is clear
that this subspace is not stable in any way\footnote{For example $R_{1,1}\subseteq f(R_0)$
but $f$ is not defined over $R_{1,1}$.}. In order to match the definition
of a piecewise affine function, we need to define $f$ over the entire space or
a stable subspace (containing the initial region). We follow this second approach
and extend the definition of $f$ on some more regions. More precisely,
we need to define $f$ over $R_i= R_{i,0}\cup R_{i,1}\cup R_{i,2}\cup
R_{i,3}\cup R_{i,41
}$
and at the moment we have only defined $f$ over $R_{i,1}=R_{i,0^\star}$ and $R_{i,4}=R_{i,1^\star}$.
Also note that $R_{i,4}=R_{i,4}^{lin}\cup R_{i,4}^{sat}$ and we define $f$ separately
on those two subregions.
In order to correctly and continuously extend $f$, we will need to further split the
region $R_{i,3}$ into linear and saturated parts
$R_{i,3}^{lin}$ and $R_{i,3}^{sat}$: see \figref{fig:region_zoom}.

Before jumping into the technical details of the extension, we start with the intuition.
First, it is crucial to understand that our main constraint is continutity: since
we already defined $f$ over $R_{i,1}$ and $R_{i,4}$, our extension need to agree
with $f$ on the borders of those regions. Furthermore, $f$ still needs to be affine,
leaving us with little flexibility. Second, the behaviour of $f$ on those regions
needs to be carefully chosen. Indeed, as we mentioned before, we had to overapproximate
regions in several places and our simulation now includes extra points. We do not
want these extra points to have completely unpredictable trajectories, otherwise
they might reach the final region by chance and break the reduction. Therefore,
our strategy is to define $f$ in such a way that its behaviour on those ``wrong
points'' still has a valid interpretation in the original SUBSET-SUM problem. We
detail this idea for the various region right after.

Let $(a,b)\in R_{i,0}\cup R_{i,2}\cup R_{i,3}$: intuitively, this point corresponds to a configuration
$c=(i,\sigma,\varepsilon_i,\ldots,\varepsilon)$ where $\varepsilon_i\notin\{0^\star,1^\star\}$.
We know by construction that this point
does not correspond to a proper trajectory so it is tempting to simply discard it.
A very simple way of discarding point is to send them to a point $x$ that is (i) stable by $f$ ($f(x)=x$)
and (ii) not in the accepting region. That way we trap the trajectory of invalid points
into a useless region of the space. Let us illustrate this on $R_{i,0}$:
let $(a,b)\in R_{i,0}$ and take $f_\mathcal{I}(a,b)=(a,b)$. It is now trivial
that the point is stuck in $R_{i,0}$. Unfortunately, $f_\mathcal{I}$ is not continuous
anymore. Indeed, for $(a,b)=(a,\beta^{-i})\in R_{i,0}\cap R_{i,1}$ we have a discontinuity on the first coordinate.
Indeed, $f_1(a,b)=a$ on one side but $f_1(a,b)=a+2^{-p}$ on the other. A simple fix
is to take $f(a,b)=(a+2^{-p},0)$, this corresponds to:
\[\text{\textbf{if} }\varepsilon_i=0\text{ \textbf{then} } f_\mathcal{I}\begin{pmatrix}
    0.\raisebox{-0.4ex}{\begin{tikzpicture}[scale=1]
        \fill[draw,color=blue!50!white,fill=blue!30!white] (0,0) rectangle (1.5,0.5);
        \fill[draw,color=red!50!white,fill=red!30!white] (1.5,0) rectangle (3,0.5);
        \node[inner sep=0] at (0.75,0.2) {$i$};
        \node[inner sep=0] at (2.25,0.2) {$\sigma$};
        \end{tikzpicture}}\\
    0.\raisebox{-0.4ex}{\begin{tikzpicture}[scale=1]
        \fill[draw,color=darkgreen!50!white,fill=darkgreen!30!white] (0,0) rectangle (3,0.5);
        \foreach \x in {0.5,1,1.5,2,2.5} \draw[color=darkgreen!50!white] (\x,0) -- (\x,0.5);
        \node[inner sep=0] at (0.25,0.2) {$0$};
        \node[inner sep=0] at (0.75,0.2) {$\ldots$};
        \node[inner sep=0] at (1.25,0.2) {$\varepsilon_i$};
        \node[inner sep=0] at (1.75,0.2) {$\scriptscriptstyle\varepsilon_{i+1}$};
        \node[inner sep=0] at (2.25,0.2) {$\ldots$};
        \node[inner sep=0] at (2.75,0.2) {$\varepsilon_n$};\end{tikzpicture}}
    \end{pmatrix}
    =
    \begin{pmatrix}
    0.\raisebox{-0.4ex}{\begin{tikzpicture}[scale=1]
        \fill[draw,color=blue!50!white,fill=blue!30!white] (0,0) rectangle (1.5,0.5);
        \fill[draw,color=red!50!white,fill=red!30!white] (1.5,0) rectangle (3,0.5);
        \node[inner sep=0] at (0.75,0.2) {$i+1$};
        \node[inner sep=0] at (2.25,0.2) {$\sigma$};
        \end{tikzpicture}}\\
    0.\raisebox{-0.4ex}{\begin{tikzpicture}[scale=1]
        \fill[draw,color=white,fill=white] (2.5,0) rectangle (3,0.5);
        \fill[draw,color=darkgreen!50!white,fill=darkgreen!30!white] (0,0) rectangle (3,0.5);
        \foreach \x in {0.5,1,1.5,2,2.5} \draw[color=darkgreen!50!white] (\x,0) -- (\x,0.5);
        \node[inner sep=0] at (0.25,0.2) {$0$};
        \node[inner sep=0] at (0.75,0.2) {$\ldots$};
        \node[inner sep=0] at (1.25,0.2) {$0$};
        \node[inner sep=0] at (1.75,0.2) {$0$};
        \node[inner sep=0] at (2.25,0.2) {$\ldots$};
        \node[inner sep=0] at (2.75,0.2) {$0$};\end{tikzpicture}}
    \end{pmatrix}.
\]

One can check that $f$ is also continuous on the second coordinate. The reason
why this choice is clever is because $f(R_{i,0})\subseteq R_{i+1,0}$. Although
the invalid points are not stable, they are now stuck in the bottom regions $R_{j,0}$
for the rest of the simulation.

Unfortunately, this trick now does not work on $R_{i,2}$ because of the continuity requirement
on the second coordinate. Indeed, on $R_{i,1}\cap R_{i,2}$ we have that $f(a,b)=(a+2^{-p},\beta^{-i})$.
To understand what it means, imagine a configuration where $\varepsilon_i=2$ and
all the remaining $\varepsilon_j$ are $0$, then
its image by $f$ corresponds to a configuration were all the remaining
$\varepsilon_j$ are $4=1^\star$ plus an error. Indeed, $\beta^{-i}=\sum_{n=i+1}^\infty 4\beta^{-n}$.
In other words, we have:
\[\text{\textbf{if} }\varepsilon_i=2\text{ \textbf{then} } f_\mathcal{I}\begin{pmatrix}
    0.\raisebox{-0.4ex}{\begin{tikzpicture}[scale=1]
        \fill[draw,color=blue!50!white,fill=blue!30!white] (0,0) rectangle (1.5,0.5);
        \fill[draw,color=red!50!white,fill=red!30!white] (1.5,0) rectangle (3,0.5);
        \node[inner sep=0] at (0.75,0.2) {$i$};
        \node[inner sep=0] at (2.25,0.2) {$\sigma$};
        \end{tikzpicture}}\\
    0.\raisebox{-0.4ex}{\begin{tikzpicture}[scale=1]
        \fill[draw,color=darkgreen!50!white,fill=darkgreen!30!white] (0,0) rectangle (3,0.5);
        \foreach \x in {0.5,1,1.5,2,2.5} \draw[color=darkgreen!50!white] (\x,0) -- (\x,0.5);
        \node[inner sep=0] at (0.25,0.2) {$0$};
        \node[inner sep=0] at (0.75,0.2) {$\ldots$};
        \node[inner sep=0] at (1.25,0.2) {$\varepsilon_i$};
        \node[inner sep=0] at (1.75,0.2) {$0$};
        \node[inner sep=0] at (2.25,0.2) {$\ldots$};
        \node[inner sep=0] at (2.75,0.2) {$0$};\end{tikzpicture}}
    \end{pmatrix}
    =
    \begin{pmatrix}
    0.\raisebox{-0.4ex}{\begin{tikzpicture}[scale=1]
        \fill[draw,color=blue!50!white,fill=blue!30!white] (0,0) rectangle (1.5,0.5);
        \fill[draw,color=red!50!white,fill=red!30!white] (1.5,0) rectangle (3,0.5);
        \path (3,0) rectangle (4,0.5);
        \node[inner sep=0] at (0.75,0.2) {$i+1$};
        \node[inner sep=0] at (2.25,0.2) {$\sigma$};
        \end{tikzpicture}}\\
    0.\raisebox{-0.4ex}{\begin{tikzpicture}[scale=1]
        \fill[draw,color=white,fill=white] (2.5,0) rectangle (4,0.5);
        \fill[draw,color=darkgreen!50!white,fill=darkgreen!30!white] (0,0) rectangle (4,0.5);
        \foreach \x in {0.5,1,1.5,2,2.5,3,3.5} \draw[color=darkgreen!50!white] (\x,0) -- (\x,0.5);
        \node[inner sep=0] at (0.25,0.2) {$0$};
        \node[inner sep=0] at (0.75,0.2) {$\ldots$};
        \node[inner sep=0] at (1.25,0.2) {$0$};
        \node[inner sep=0] at (1.75,0.2) {$4$};
        \node[inner sep=0] at (2.25,0.2) {$\ldots$};
        \node[inner sep=0] at (2.75,0.2) {$4$};
        \node[inner sep=0] at (3.25,0.2) {$4$};
        \node[inner sep=0] at (3.75,0.2) {$\ldots$};\end{tikzpicture}}
    \end{pmatrix}.
\]
Furthermore, thinking about the future a bit, on $R_{i,3}\cap R_{i,4}$ we have that $f(a,b)=(a+2^{-p},0)$.
In other words, somewhere on $R_{i,2}\cup R_{i,3}$, the second coordinate has to go
from $\beta^{-i}$ to $0$ in a continuous fashion. This is where the clever tricks
comes in: we can continuously change the second coordinate from $1$ to $0$ such
that the action of $f$ looks like all $\varepsilon_j$ were ``flipped'': $0$
is exchange with $4$ and $1$ with $2$. To visualise how this possible, simply
think about the configuration as having an infinite number of $\varepsilon_j$ (although
we use a finite amount of them) that gets turned into an infinite number of $\mu_j$ where
$\mu_j=4-\varepsilon_j$:
\[f_\mathcal{I}\begin{pmatrix}
    0.\raisebox{-0.4ex}{\begin{tikzpicture}[scale=1]
        \fill[draw,color=blue!50!white,fill=blue!30!white] (0,0) rectangle (1.5,0.5);
        \fill[draw,color=red!50!white,fill=red!30!white] (1.5,0) rectangle (3,0.5);
        \path (3,0) rectangle (4,0.5);
        \node[inner sep=0] at (0.75,0.2) {$i+1$};
        \node[inner sep=0] at (2.25,0.2) {$\sigma$};
        \end{tikzpicture}}\\
    0.\raisebox{-0.4ex}{\begin{tikzpicture}[scale=1]
        \fill[draw,color=white,fill=white] (2.5,0) rectangle (4,0.5);
        \fill[draw,color=darkgreen!50!white,fill=darkgreen!30!white] (0,0) rectangle (4,0.5);
        \foreach \x in {0.5,1,1.5,2,2.5,3,3.5} \draw[color=darkgreen!50!white] (\x,0) -- (\x,0.5);
        \node[inner sep=0] at (0.25,0.2) {$0$};
        \node[inner sep=0] at (0.75,0.2) {$\ldots$};
        \node[inner sep=0] at (1.25,0.2) {$2$};
        \node[inner sep=0] at (1.75,0.2) {$\scriptstyle\varepsilon_{i+1}$};
        \node[inner sep=0] at (2.25,0.2) {$\ldots$};
        \node[inner sep=0] at (2.75,0.2) {$\varepsilon_{n}$};
        \node[inner sep=0] at (3.3,0.2) {$\scriptstyle\varepsilon_{n+1}$};
        \node[inner sep=0] at (3.75,0.2) {$\ldots$};\end{tikzpicture}}
    \end{pmatrix}
    =
    \begin{pmatrix}
    0.\raisebox{-0.4ex}{\begin{tikzpicture}[scale=1]
        \fill[draw,color=blue!50!white,fill=blue!30!white] (0,0) rectangle (1.5,0.5);
        \fill[draw,color=red!50!white,fill=red!30!white] (1.5,0) rectangle (3,0.5);
        \path (3,0) rectangle (4,0.5);
        \node[inner sep=0] at (0.75,0.2) {$i+1$};
        \node[inner sep=0] at (2.25,0.2) {$\sigma$};
        \end{tikzpicture}}\\
    0.\raisebox{-0.4ex}{\begin{tikzpicture}[scale=1]
        \fill[draw,color=white,fill=white] (2.5,0) rectangle (4,0.5);
        \fill[draw,color=darkgreen!50!white,fill=darkgreen!30!white] (0,0) rectangle (4,0.5);
        \foreach \x in {0.5,1,1.5,2,2.5,3,3.5} \draw[color=darkgreen!50!white] (\x,0) -- (\x,0.5);
        \node[inner sep=0] at (0.25,0.2) {$0$};
        \node[inner sep=0] at (0.75,0.2) {$\ldots$};
        \node[inner sep=0] at (1.25,0.2) {$0$};
        \node[inner sep=0] at (1.75,0.2) {$\scriptstyle\mu_{i+1}$};
        \node[inner sep=0] at (2.25,0.2) {$\ldots$};
        \node[inner sep=0] at (2.75,0.2) {$\mu_{n}$};
        \node[inner sep=0] at (3.3,0.2) {$\scriptstyle\mu_{n+1}$};
        \node[inner sep=0] at (3.75,0.2) {$\ldots$};\end{tikzpicture}}
    \end{pmatrix}.
\]
This might seem convoluted at first until one realises why this is helpful. Imagine
an extended SUBSET-SUM simulation where we now have three actions instead of two:
\begin{itemize}
\item $\varepsilon_i=0$: go to next number,
\item $\varepsilon_i=1$: add $A_i$ and go to next number,
\item $\varepsilon_i=2$: flip all remaining $\varepsilon_j$ and go to next number.
\end{itemize}
Our simulation corresponds to this extended problem, which is still NP-complete.
The remaining region $R_{i,3}$ follows the same principle as $R_{i,0}$, with
a slight complication because of the saturated sum. \figref{fig:region_zoom_informal}
gives the interpretation of the definition of $f$ on each subregion.

\begin{definition}[Extended region splitting]\label{def:ext_regions}
For $i\in\{1,\ldots,n\}$ and $\alpha\in\{0,\ldots,\beta-1\}$, define:
\[R_{i,3}^{lin}=R_{i,3}\cap\left\{(a,b)\thinspace\big|\thinspace
    b\beta^{i}-3\leqslant \frac{2^{-p-1}+i2^{-p}-a}{2^{-p-1}-(B+1-A_i)2^{-q}}\right\}
\qquad R_{i,3}^{sat}=R_{i,3}\setminus R_{i,3}^{lin}\]
\end{definition}

It should be clear by definition that $R_{i,3}^{sat}=R_{i,3}^{lin}\cup
R_{i,3}^{sat}$ and that the two subregions are disjoint except on the border.

\begin{definition}[Extended piecewise affine simulation]\label{def:subset_paf_ex}
\[
f_\mathcal{I}(a,b)=\begin{cases}
(a+2^{-p},0)&\text{if }(a,b)\in R_{i,0}\\
(a+2^{-p},3\beta^{-i}-b)&\text{if }(a,b)\in R_{i,2}\\
(a+2^{-p}+A_i2^{-q}(b\beta^{i}-3),0)&\text{if }(a,b)\in R_{i,3}^{lin}\\
((i+\frac{3}{2})2^{-p}-(b\beta^i-3)(2^{-p-1}-(B+1)2^{-q}),0)&\text{if }(a,b)\in R_{i,3}^{sat}\\
\end{cases}
\]
\end{definition}

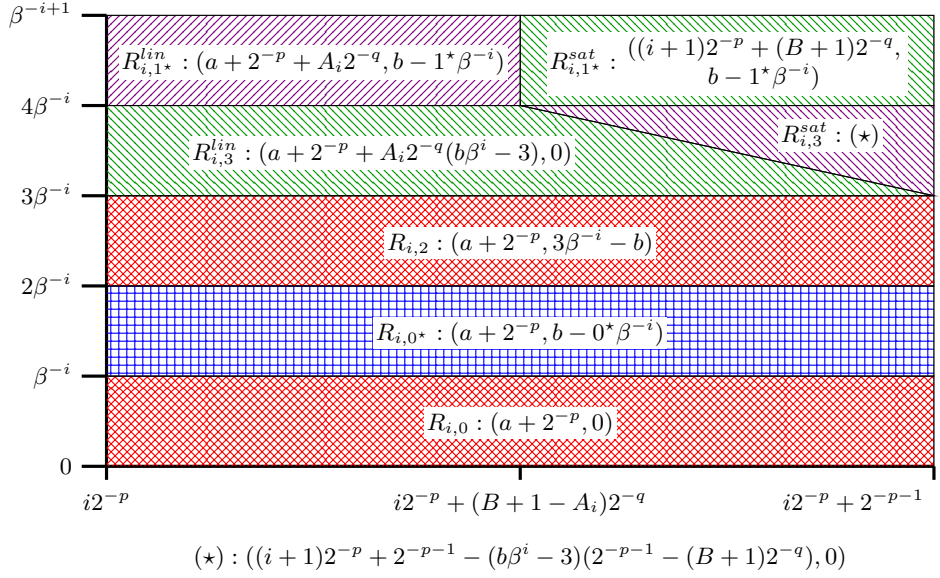
\begin{figure}
\begin{center}
\begin{tikzpicture}[xscale=11,yscale=6]
  \pgfmathsetmacro{\n}{2}
  \pgfmathsetmacro{\maxn}{\n+1}
  \pgfmathsetmacro{\p}{ceil(log2(\n+2))}
  \pgfmathsetmacro{\B}{2}
  \pgfmathsetmacro{\vomega}{ceil(log2(\B+2))}
  \pgfmathsetmacro{\q}{\p+\vomega+1}
  \pgfmathsetmacro{\vbeta}{5}
  \pgfmathsetmacro{\xstep}{1}
  \pgfmathsetmacro{\ystep}{pow(\vbeta,-1)}
  \pgfmathsetmacro{\xsat}{0.5}
  \draw[very thin,color=gray,xstep=\xstep,ystep=\ystep] (0.0,0.0) grid (1,1);
  \draw[very thick] (0,0) -- (1, 0);
  \draw[very thick] (0,0) -- (0, 1);
  \draw[very thick] (0,0) -- (0,-0.03) node[below] {$i2^{-p}$};
  \draw[very thick] (1,0) -- (1,-0.03) node[below left] {$i2^{-p}+2^{-p-1}$};
  \draw[very thick] (0,1) -- (-0.03,1) node[left] {$\beta^{-i+1}$};
  \draw[very thick] (0,0) -- (-0.03,0) node[left] {$0$};
  \draw[very thick] (\xsat,0) -- (\xsat,-0.03) node[below] {$i2^{-p}+(B+1-A_i)2^{-q}$};
  \colorlet{mygreen}{green!60!black}
  \colorlet{mymagenta}{blue!50!red}
  \draw[pattern color=red,pattern=crosshatch] (0,0) rectangle (1,\ystep);
  \draw (.5,\ystep/2-0.02) node[anchor=base,fill=white,inner sep=1pt]{$R_{i,0}:(a+2^{-p},0)$};
  \draw[pattern color=blue,pattern=grid] (0,\ystep) rectangle (1,2*\ystep);
  \draw (.5,\ystep+\ystep/2-0.02) node[anchor=base,fill=white,inner sep=1pt]{
      $R_{i,0^\star}:(a+2^{-p},b-0^\star\beta^{-i})$};
  \draw[pattern color=red,pattern=crosshatch] (0,2*\ystep) rectangle (1,3*\ystep);
  \draw (.5,2*\ystep+\ystep/2-0.02) node[anchor=base,fill=white,inner sep=1pt]{$R_{i,2}:(a+2^{-p},3\beta^{-i}-b)$};
  \draw[pattern color=mygreen,pattern=north west lines] (0,3*\ystep) -- (1,3*\ystep) -- (\xsat,4*\ystep) -- (0,4*\ystep) -- (0,3*\ystep);
  \draw (2*\xsat/3,3*\ystep+\ystep/2-0.02) node[anchor=base,fill=white,inner sep=1pt]{
      $R_{i,3}^{lin}:(a+2^{-p}+A_i2^{-q}(b\beta^i-3),0)$};
  \draw[pattern color=mymagenta,pattern=north west lines] (1,3*\ystep) -- (1,4*\ystep) -- (\xsat,4*\ystep) -- (1,3*\ystep);
  \draw (\xsat/4+3/4.,3*\ystep+\ystep/2+0.02) node[anchor=base,fill=white,inner sep=1pt]{
      $R_{i,3}^{sat}:(\star)$}; 
  \draw[pattern color=mymagenta,pattern=north east lines] (0,4*\ystep) rectangle (\xsat,5*\ystep);
  \draw (\xsat/2,4*\ystep+\ystep/2-0.02) node[anchor=base,fill=white,inner sep=1pt]{
      $R_{i,1^\star}^{lin}:(a+2^{-p}+A_i2^{-q},b-1^\star\beta^{-i})$};
  \draw[pattern color=mygreen,pattern=north west lines] (\xsat,4*\ystep) rectangle (1,5*\ystep);
  \draw (.5+\xsat/2,4*\ystep+\ystep/2-0.02) node[anchor=base,fill=white,inner sep=1pt]{
      $R_{i,1^\star}^{sat}:\begin{array}{c}((i+1)2^{-p}+(B+1)2^{-q},\\b-1^\star\beta^{-i})\end{array}$};

  \draw[very thick] (0,\ystep) -- (-0.03,\ystep) node[left] {$\beta^{-i}$};
  \foreach \valpha in {2,...,4}
  {
      \pgfmathsetmacro{\y}{\valpha*\ystep};
      \draw[very thick] (0,\y) -- (-0.03,\y) node[left] {$\valpha\beta^{-i}$};
  }
  \draw (.5,-0.2) node {$(\star):((i+1)2^{-p}+2^{-p-1}-(b\beta^i-3)(2^{-p-1}-(B+1)2^{-q}),0)$};
\end{tikzpicture}
\end{center}
\caption{Zoom on one $R_i$ with the subregions and formulas.\label{fig:region_zoom}}
\end{figure}

\begin{figure}
\begin{center}
\begin{tikzpicture}[xscale=11,yscale=6]
  \pgfmathsetmacro{\n}{2}
  \pgfmathsetmacro{\maxn}{\n+1}
  \pgfmathsetmacro{\p}{ceil(log2(\n+2))}
  \pgfmathsetmacro{\B}{2}
  \pgfmathsetmacro{\vomega}{ceil(log2(\B+2))}
  \pgfmathsetmacro{\q}{\p+\vomega+1}
  \pgfmathsetmacro{\vbeta}{5}
  \pgfmathsetmacro{\xstep}{1}
  \pgfmathsetmacro{\ystep}{pow(\vbeta,-1)}
  \pgfmathsetmacro{\xsat}{0.5}
  \draw[very thin,color=gray,xstep=\xstep,ystep=\ystep] (0.0,0.0) grid (1,1);
  \draw[very thick] (0,0) -- (1, 0);
  \draw[very thick] (0,0) -- (0, 1);
  \draw[very thick] (0,0) -- (0,-0.03) node[below] {$i2^{-p}$};
  \draw[very thick] (1,0) -- (1,-0.03) node[below left] {$i2^{-p}+2^{-p-1}$};
  \draw[very thick] (0,1) -- (-0.03,1) node[left] {$\beta^{-i+1}$};
  \draw[very thick] (0,0) -- (-0.03,0) node[left] {$0$};
  \draw[very thick] (\xsat,0) -- (\xsat,-0.03) node[below] {$i2^{-p}+(B+1-A_i)2^{-q}$};
  \colorlet{mygreen}{green!60!black}
  \colorlet{mymagenta}{blue!50!red}
  \draw[pattern color=red,pattern=crosshatch] (0,0) rectangle (1,\ystep);
  \draw (.5,\ystep/2-0.02) node[anchor=base,fill=white,inner sep=1pt]{
    $R_{i,0}:\text{set all remaining $\varepsilon_j$ to $0$}$};
  \draw[pattern color=blue,pattern=grid] (0,\ystep) rectangle (1,2*\ystep);
  \draw (.5,\ystep+\ystep/2-0.02) node[anchor=base,fill=white,inner sep=1pt]{
      $R_{i,0^\star}:\text{go to next number}$};
  \draw[pattern color=red,pattern=crosshatch] (0,2*\ystep) rectangle (1,3*\ystep);
  \draw (.5,2*\ystep+\ystep/2-0.02) node[anchor=base,fill=white,inner sep=1pt]{
    $R_{i,2}:\text{flip all remaining $\varepsilon_j$}$};
  \draw[pattern color=mygreen,pattern=north west lines] (0,3*\ystep) -- (1,3*\ystep) -- (\xsat,4*\ystep) -- (0,4*\ystep) -- (0,3*\ystep);
  \draw (1.8*\xsat/3,3*\ystep+\ystep/2-0.02) node[anchor=base,fill=white,inner sep=1pt]{
      $R_{i,3}^{lin}:\left\{\text{\begin{tabular}{l@{}}set all remaining $\varepsilon_j$ to $0$\\
        add a funny amount\textsuperscript{$\dagger$}\end{tabular}}\right.$};
  \draw[pattern color=mymagenta,pattern=north west lines] (1,3*\ystep) -- (1,4*\ystep) -- (\xsat,4*\ystep) -- (1,3*\ystep);
  \draw (\xsat/4+3/4.,3*\ystep+\ystep/2+0.02) node[anchor=base,fill=white,inner sep=1pt]{
      $R_{i,3}^{sat}:(\star)$}; 
  \draw[pattern color=mymagenta,pattern=north east lines] (0,4*\ystep) rectangle (\xsat,5*\ystep);
  \draw (\xsat/2,4*\ystep+\ystep/2-0.02) node[anchor=base,fill=white,inner sep=1pt]{
      $R_{i,1^\star}^{lin}:\text{add $A_i$}$};
  \draw[pattern color=mygreen,pattern=north west lines] (\xsat,4*\ystep) rectangle (1,5*\ystep);
  \draw (.5+\xsat/2,4*\ystep+\ystep/2-0.02) node[anchor=base,fill=white,inner sep=1pt]{
      $R_{i,1^\star}^{sat}:\text{set sum to $B+1$}$};

  \draw[very thick] (0,\ystep) -- (-0.03,\ystep) node[left] {$\beta^{-i}$};
  \foreach \valpha in {2,...,4}
  {
      \pgfmathsetmacro{\y}{\valpha*\ystep};
      \draw[very thick] (0,\y) -- (-0.03,\y) node[left] {$\valpha\beta^{-i}$};
  }
  \draw (0,-0.2) node[anchor=west] {$(\star):\left\{\text{\begin{tabular}{l@{}}set all remaining $\varepsilon_j$ to $0$\\
  add a funny amount\textsuperscript{$\dagger$}\end{tabular}}\right.$};
  \draw (.5,-0.2) node[anchor=west] {
    \begin{tabular}{@{}l@{}}{}\textsuperscript{$\dagger$}\thinspace The amount depends on $(a,b)$ and\\
        has no simple interpretation.\end{tabular}};
\end{tikzpicture}
\end{center}
\caption{Interpretation of the behaviour of $f$ on each suregion. See also \figref{fig:region_zoom}
and Section~\ref{sec:subset_sum_affine}. All regions include an implicit ``go to next number''.\label{fig:region_zoom_informal}}
\end{figure}
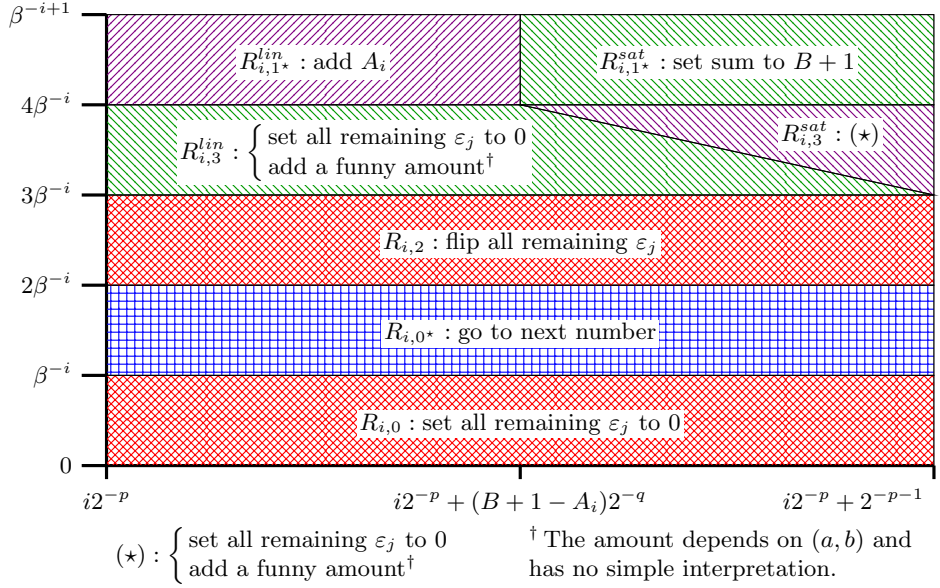

This extension was carefully chosen for its properties. In particular, we will
see that $f$ is still continuous, Also, the domain of definition of $f$ is $f$-stable
(i.e. $f(\dom f)\subseteq\dom f$).
And finally, we will see that $f$ is somehow ``reversible''.

\begin{lemma}[Simulation is continuous]\label{lem:simul_cont}
For any $i\in\{1,\ldots,n\}$, $f_\mathcal{I}(R_i)$ is well-defined and continuous over $R_i$.
\end{lemma}

\begin{proof}
As outlined on \figref{fig:region_zoom}, we need to check that the definitions of $f$
match at the borders of each subregions of $R_i$. More precisely, we need to check
that \defref{def:subset_paf} and \defref{def:subset_paf_ex} agree on all borders.
\begin{itemize}
\item $(a,b)\in R_{i,0}\cap R_{i,0^\star}$: the first component is computed using the same
    formula so is clearly continuous, the second component is always $0$ on both side of the border
    because $b-0^\star\beta^{-i}=0$ for $b=\beta^{-i}$
\item $(a,b)\in R_{i,0^\star}\cap R_{i,2}$: the first component is computed using the same
    formula so is clearly continuous, the second component is always $\beta^{-i}$ on both side of the border
    because $b-0^\star\beta^{-i}=3\beta^{-i}-b=\beta^{-i}$ for $b=2\beta^{-i}$
\item $(a,b)\in R_{i,2}\cap R_{i,3}^{lin}$: the first component is $a+2^{-p}$ and 
    the second component is $0$ on both side of the border because $3\beta^{-i}-b=b\beta^{i}-3=0$ for $b=3\beta^{-i}$
\item $(a,b)\in R_{i,3}^{lin}\cap R_{i,1^\star}^{lin}$: the first component is $a+2^{-p}+A_i2^{-q}$ 
    and the second component is $0$ on both side of the border because $b\beta^{i}-3=1$ and 
    $b-1^\star\beta^{-i}=0$ for $b=4\beta^{-i}$
\item $(a,b)\in R_{i,3}^{lin}\cap R_{i,3}^{sat}$: the second component is always $0$ on both regions
    so is clearly continuous. From \defref{def:ext_regions} one can see
    that $b\beta^i-3=\frac{Y}{X}$ holds on the border where $Y=2^{-p-1}+i2^{-p}-a$ and
    $X=2^{-p-1}-(B+1-A_i)2^{-q}$. Consequently,
    if we compute the difference between the two expression at the borders, we get:
    \begin{align*}
    &\left((i+1)2^{-p}+2^{-p-1}-(b\beta^i-3)(2^{-p-1}-(B+1)2^{-q})\right)\\
    &-\left(a+2^{-p}+A_i2^{-q}(b\beta^i-3)\right)\\
    &=i2^{-p}+2^{-p-1}-a-\frac{Y}{X}(2^{-p-1}-(B+1)2^{-q}+A_i2^{-q})\\
    &=Y-\frac{Y}{X}X=0
    \end{align*}
    this proves that they are equal.
\item $(a,b)\in R_{i,3}^{sat}\cap R_{i,1^\star}^{sat}$: the first component is $(i+1)2^{-p}+(B+1)2^{-q}$
    and the second component is $0$ on both side of the border because $b\beta^i-3=1$ and $b-1^\star\beta^{-i}=0$ 
    for $\beta=4\beta^{-i}$
\item $(a,b)\in R_{i,1^\star}^{lin}\cap R_{i,1^\star}^{sat}$: the first component is $(i+1)2^{-p}+(B+1)2^{-q}$
    on both side of the border because $a=i2^{-p}+(B+1-A_i)2^{-q}$, and the second component is computed using
    the same formula so is clearly continuous
\end{itemize}
\qed\end{proof}

\begin{lemma}[Simulation is stable]\label{lem:simul_stable}
For any $i\in\{1,\ldots,n\}$, $f_\mathcal{I}(R_i)\subseteq R_{i+1}$. Furthermore,
$f(R_0)\subseteq R_1$ and $f(R_{n+1})\subseteq R_{n+1}$.
\end{lemma}

\begin{proof}
We need to examinate all possible cases for $(a,b)\in R_i$. Since $R_i=\bigcup_{\alpha=0}^{4}R_{i,\alpha}$
and that $R_{i,\alpha}=R_{i,\alpha}^{lin}\cup R_{i,\alpha}^{sat}$ we indeed cover all cases.
\begin{itemize}
\item If $(a,b)\in R_0$: then $f_\mathcal{I}(a,b)=(a+2^{-p},b)$ so 
$f_\mathcal{I}(R_0)=f_\mathcal{I}([0,2^{-p-1}]\times[0,1])=[2^{-p},2^{-p}+2^{-p-1}]\times[0,1]=R_1$.
\item If $(a,b)\in R_{n+1}$: then $f_\mathcal{I}(a,b)=(a,b)$ so $f_\mathcal{I}(R_{n+1})=R_{n+1}$.
\item If $(a,b)\in R_{i,0}$: then $f_\mathcal{I}(a,b)=(a+2^{-p},0)$ so
$f_\mathcal{I}(R_{i,0})=f_\mathcal{I}([i2^{-p},i2^{-p}+2^{-p-1}]\times[0,\beta^{-i}])
=[(i+1)2^{-p},(i+1)2^{-p}+2^{-p-1}]\times\{0\}\subseteq R_{i+1}$.
\item If $(a,b)\in R_{i,1}=R_{i,0^\star}$: then $f_\mathcal{I}(a,b)=(a+2^{-p},b-0^\star\beta^{-i})$ so
$f_\mathcal{I}(R_{i,1})=f_\mathcal{I}([i2^{-p},i2^{-p}+2^{-p-1}]\times[\beta^{-i},2\beta^{-i}])
=[(i+1)2^{-p},(i+1)2^{-p}+2^{-p-1}]\times[0,\beta^{-i}]=R_{i+1}$.
\item If $(a,b)\in R_{i,2}$: then $f_\mathcal{I}(a,b)=(a+2^{-p},3\beta^{-i}-b)$ so
$f_\mathcal{I}(R_{i,2})=f_\mathcal{I}([i2^{-p},i2^{-p}+2^{-p-1}]\times[2\beta^{-i},3\beta^{-i}])
=[(i+1)2^{-p},(i+1)2^{-p}+2^{-p-1}]\times[0,\beta^{-i}]=R_{i+1}$.
\item If $(a,b)\in R_{i,3}^{lin}$: the image of the second component is always $0$ so it's easy for this one,
also from \defref{def:ext_regions}, $b\beta^i-3\leqslant\frac{2^{-p-1}+i2^{-p}-a}{2^{-p-1}-(B+1-A_i)2^{-q}}
\leqslant \frac{2^{-p-1}+i2^{-p}-a}{A_i2^{-q}}$ because $2^{-p-1}-(B+1)2^{-q}\geqslant0$ 
since $(B+1)2^{-q}\leqslant2^\omega2^{-q}\leqslant2^{-p-1}$. Consequently, for the first coordinate we get that
$f_\mathcal{I}(a,b)_1\leqslant a+2^{-p}+A_i2^{-q}\frac{2^{-p-1}+i2^{-p}-a}{A_i2^{-q}}\leqslant(i+1)2^{-p}+2^{-p-1}$.
Also, since $i2^{-p}\leqslant a\leqslant i2^{-p}+2^{-p-1}$, it is clear that $f_\mathcal{I}(a,b)_1\geqslant (i+1)2^{-p}$.
So finally, $f_\mathcal{I}(R_{i,3}^{lin})\subseteq [(i+1)2^{-p},(i+1)2^{-p}+2^{-p-1}]\times\{0\}\subset R_{i+1}$.
\item If $(a,b)\in R_{i,3}^{sat}$: the image of the second component is always $0$ so it's easy for this one,
also from \defref{def:ext_regions}, $b\beta^i-3\geqslant\frac{2^{-p-1}+i2^{-p}-a}{2^{-p-1}-(B+1-A_i)2^{-q}}
\geqslant\frac{2^{-p-1}+i2^{-p}-a}{2^{-p-1}-(B+1)2^{-q}}$ because $A_i\geqslant0$. Consequently, for the first coordinate we get that
$f_\mathcal{I}(a,b)_1\leqslant (i+1)2^{-p}+2^{-p-1}-(2^{-p-1}-(B+1)2^{-q})\frac{2^{-p-1}+i2^{-p}-a}{2^{-p-1}-(B+1)2^{-q}}
\leqslant (i+1)2^{-p}+2^{-p-1}+i2^{-p}+2^{-p-1}-a\leqslant (i+1)2^{-p}+2^{-p-1}$ since $a\leqslant i2^{-p}+2^{-p-1}$.
Also since $b\beta^i-3\leqslant1$ we get that 
$f_\mathcal{I}(a,b)_1\geqslant (i+1)2^{-p}+2^{-p-1}-(2^{-p-1}-(B+1)2^{-q})\times 1\geqslant (i+1)2^{-p}+(B+1)2^{-q}$.
So finally, $f_\mathcal{I}(R_{i,3}^{sat})\subseteq [(i+1)2^{-p}+(B+1)2^{-q},(i+1)2^{-p}+2^{-p-1}]\times\{0\}\subset R_{i+1}$.
\item If $(a,b)\in R_{i,4}^{lin}=R_{i,1^\star}^{lin}$: then $f_\mathcal{I}(a,b)=(a+2^{-p}+A_i2^{-q},b-1^\star\beta^{-i})$ so
$f_\mathcal{I}(R_{i,4}^{lin})=f_\mathcal{I}([i2^{-p},i2^{-p}+(B+1-A_i)2^{-q}]\times[4\beta^{-i},5\beta^{-i}])
=[(i+1)2^{-p}+A_i2^{-q},(i+1)2^{-p}+(B+1)2^{-q}]\times[0,\beta^{-i}]\subseteq R_{i+1}$ because $(B+1)2^{-q}\leqslant 2^{-p-1}$.
\item If $(a,b)\in R_{i,4}^{sat}$: then $f_\mathcal{I}(a,b)=((i+1)2^{-p}+(B+1)2^{-q},0)$ so
$f_\mathcal{I}(R_{i,4}^{sat})=\{(i+1)2^{-p}+(B+1)2^{-q}\}\times\{0\}\subseteq R_{i+1}$.
\end{itemize}
\qed\end{proof}

We now get to the core lemma of the simulation. Up to this point, we were only
interested in forward simulation: that is given a point, what are the iterates of $x$.
In order to prove the NP-hardness result, we need a backward result: given a point,
what are the possible preimages of it. To this end, we introduce new subregions
$R_i^{unsat}$ of the $R_i$, that we call \emph{unsaturated}. Intuitively, $R_i^{unsat}$ corresponds
to the encodings where $\sigma\leqslant B$, that is the sum did not saturate at $B+1$.
We also introduce the $R_{fin}$ region, that will be the region to reach. We will
be interested in the preimages of $R_{fin}$.

\begin{definition}[Unsaturated regions]
For $i\in\{1,\ldots,n+1\}$, define
\[R_i^{unsat}=[i2^{-p},i2^{-p}+B2^{-q}]\times[\beta^{-n-1},\beta^{-i+1}-\beta^{-n-1}]\]
\[R_{fin}=[(n+1)2^{-p}+B2^{-q}-2^{-q-1},(n+1)2^{-p}+B2^{-q}]\times[\beta^{-n-1},2\beta^{-n-1}]\]
\end{definition}

\begin{lemma}[Simulation is reversible]\label{lem:simul_rev}
Let $i\in\{2,\ldots,n\}$ and $(a,b)\in R_i^{unsat}$
Then the only points $\vec{x}$ such that $f_\mathcal{I}(\vec{x})=(a,b)$ are:
\begin{itemize}
\item $\vec{x}=(a-2^{-p},b+0^\star\beta^{-i+1})\in R_{i-1,0^\star}\cap R_{i-1}^{unsat}$
\item $\vec{x}=(a-2^{-p},\beta^i-b+0^\star\beta^{-i+1})\in R_{i-1,2}\cap R_{i-1}^{unsat}$
\item $\vec{x}=(a-2^{-p}-A_i2^{-q},b+1^\star\beta^{-i+1})\in R_{i-1,1^\star}^{lin}\cap R_{i-1}^{unsat}$
(only if $a\geqslant 2^{-p}+A_i2^{-q}$)
\end{itemize}
\end{lemma}

\begin{proof}

First notice that since $f_\mathcal{I}(R_i)\subseteq R_{i+1}$ for all $i\in\{0,\ldots,n\}$,
the only candidates for $\vec{x}$ must belong to $R_{i-1}$. Furthermore, on each affine
region, there can only be one candidate except if the function is trivial.

A close look at the proof of \lemref{lem:simul_stable}
reveals that:
\begin{itemize}
\item $f_\mathcal{I}(R_{i-1,0})\subseteq [(i+1)2^{-p},(i+1)2^{-p}+2^{-p-1}]\times\{0\}$ 
shareing no point with $R_i^{unsat}$ so there is no possible candidate
\item $f_\mathcal{I}(R_{i-1,1})=R_i$ and there is only one possible candidate
\item $f_\mathcal{I}(R_{i-1,2})=R_i$ and there is only one possible candidate
\item $f_\mathcal{I}(R_{i-1,3})\subseteq [(i+1)2^{-p},(i+1)2^{-p}+2^{-p-1}]\times\{0\}$ so like
$R_{i-1,0}$ there is no possible candidate
\item $f_\mathcal{I}(R_{i-1,4}^{lin})\subseteq R_i$ and there is only one possible candidate
\item $f_\mathcal{I}(R_{i-1,4}^{sat})\subseteq [(i+1)2^{-p}+(B+1)2^{-q},(i+1)2^{-p}+2^{-p-1}]\times[0,\beta^{-i}]$
sharing no point with $R_i^{unsat}$ so there is no possible candidate
\end{itemize}

It is then only a matter of checking that the claimed formulas work and they trivially do
except for the case of $R_{i-1,4}^{lin}$ where we need the potential candidate to
belong to the region.
\qed\end{proof}

The goal of those results is to show that if there is a point in $R_{fin}$ that is reachable from $R_0$ then we
can extract, from its trajectory, a configuration that also reaches $R_{fin}$.
Furthermore, we arranged so that $R_{fin}$ contains the encoding of only one
configuration: $(n+1,B)$ (see \lemref{lem:subset_sum_simul}).

\begin{lemma}[Backward-forward identity]\label{lem:backward_forward}
For any point $\vec{x}\in R_{fin}$, if there exists
a point $\vec{y}\in R_0$ and an integer $k$ such that $f_\mathcal{I}^{[k]}(\vec{y})=\vec{x}$ then
there exists a configuration $c=(1,0,\varepsilon_1,\ldots,\varepsilon_n)$ such that 
$f_\mathcal{I}^{[k]}(\enc{c})\in R_{fin}$.
\end{lemma}

\begin{proof}

Define $\vec{y}_0=\vec{y}$ and $\vec{y}_{i+1}=f_\mathcal{I}(\vec{y}_i)$ for all $i\in\{0,k-1\}$.
Since $\vec{y}_0\in R_0$, we immediately get that $\vec{y}_i\in R_i$ using \lemref{lem:simul_stable}
and in particular, $k\geqslant n+1$ because $\vec{y}_{k}=x\in R_{n+1,0^\star}$.

Now apply \lemref{lem:simul_rev} starting from $\vec{y}_{n+1}\in R_{n+1}^{unsat}$:
we conclude that for all $i\geqslant 1$, 
$\vec{y}_i\in (R_{i,1}\cup R_{i,2}\cup R_{i,4}^{lin})\cap R_n^{unsat}$.
Define $\varepsilon_i=0$ if $\vec{y}_i\in R_{i,1}\cup R_{i,2}$ and $1$ if $\vec{y}_i\in R_{i,4}^{lin}$.
Write $\vec{y}_i=(a_i,b_i)$. Again using \lemref{lem:simul_stable}
we get that $a_{i-1}=a_{i}-2^{-p}-\varepsilon_iA_i2^{-q}$ (just check all three cases).
Also since $\vec{x}=\vec{y_{n+1}}\in R_{fin}$ then
$a_{n+1}\in[(n+1)2^{-p}+B2^{-q}-2^{-q-1},(n+1)2^{-p}+B2^{-q}]$.
Finally, $\vec{y}_0\in R_0$ so $f_\mathcal{I}(a_0,b_0)=(2^{-p},0)=(a_1,b_1)$.
We conclude that $a_1=2^{-p}$. Putting everything together we get:
\[\left\{\begin{array}{r@{}l}
a_{n+1}&=(n+1)2^{-p}+2^{-q}\sum_{i=1}^{n}\varepsilon_iA_i\\
a_{n+1}&\in [(n+1)2^{-p}+B2^{-q}-2^{-q-1},(n+1)2^{-p}+B2^{-q}]
\end{array}\right.\]
Since the $A_i$, $B$ are integers and $\varepsilon_i\in\{0,1\}$, we get that
$B=\sum_{i=1}^{n}A_i\varepsilon_i$.
Apply \lemref{lem:simul_correct} on the configuration to conclude.
\qed\end{proof}

\begin{lemma}[Final region is accepting]\label{lem:fin_is_accept}
For any configuration $c$, if $\enc{c}\in R_{fin}$ then $c=(n+1,B)$.
\end{lemma}

\begin{proof}
Write $c=(i,\sigma,\varepsilon_i,\ldots,\varepsilon_n)$, then$\enc{c}=
\left(i2^{-p}+\sigma2^{-q},\sum_{j=i}^n\varepsilon_i^\star\beta^{-i}+0^\star\beta^{-n-1}\right)$.
It implies that $i2^{-p}+\sigma2^{-q}\in[(n+1)2^{-p}+B2^{-q}-2^{-q-1},(n+1)2^{-p}+B2^{-q}]$
and because $i$ is an integer in range $[0,n+1]$ and $\sigma$ an integer in range $[0,B+1]$,
necessarily $i=n+1$ and $\sigma=B$.
\qed\end{proof}

\subsection{Complexity result}

We now have all the tools to show that \texttt{REACH-REGION-TIME} is an NP-hard problem.

\begin{theorem}\label{th:reach_np_hard}
\texttt{REACH-REGION-TIME} is NP-hard for $d\geqslant2$.
\end{theorem}

\begin{proof}
Let $\mathcal{I}=(B,A_1,\ldots,A_n)$ be an instance of \texttt{SUBSET-SUM}. We consider the instance
$\mathcal{J}$ of \texttt{REACH-REGION-TIME} defined in the previous section with maximum
number of iterations set to $n$ (the number of $A_i$),
the initial region set to $R_0$ and the final region set to $R_{fin}$.
One easily checks that this instance has polynomial size in the size of $\mathcal{I}$.
The two directions of the proof are:
\begin{itemize}
\item If $\mathcal{I}$ is satisfiable then use \lemref{lem:subset_sum_it} and
\lemref{lem:simul_correct} to conclude that there is a point $x\in R_0$ in the initial region
such that $f_\mathcal{I}^{[n]}(x)\in R_{fin}$ so $\mathcal{J}$ is satisfiable.
\item If $\mathcal{J}$ is satisfiable then there exists $x\in R_0$ and $k\leqslant n$ such
that $f_\mathcal{I}^{[k]}(x)\in R_{fin}$. Use \lemref{lem:backward_forward} and \lemref{lem:simul_correct} to conclude
that there exists a configuration $c=(1,0,\varepsilon_1,\ldots,\varepsilon_n)$ such
that $\enc{T_\mathcal{I}^{[k]}(c)}=f_\mathcal{I}^{[k]}(\enc{c})\in R_{fin}$. Apply
\lemref{lem:fin_is_accept} and use the injectivity of the encoding to conclude that
$T_\mathcal{I}^{[k]}(c)=(n+1,B)$ and \lemref{lem:subset_sum_simul} to get that $\mathcal{I}$ is satisfiable.
\end{itemize}
\qed\end{proof}

\section{Bounded Time Reachability is in NP}\label{sec:paf:solving}

In the previous section we have
shown that the \texttt{REACH-REGION-TIME} problem is NP-hard. We now give a
more precise characterization of the complexity of this problem, by
proving that it is NP-complete. Since we have shown its NP-hardness,
the only thing that remains to be shown is that \texttt{REACH-REGION-TIME}
belongs to NP. This is done in this section.

\subsection{Notations and definitions}

For any $i\in \{1,\ldots,d\}$, let $\pi_i^d \colon \R^d \to \R$ denote the $i^{th}$ projection function,
that is, $\pi(x_1,\ldots,x_d) = x_i$. Let $g_d\colon\R^{d+1}\to\R^d$ be defined by 
$g_d(x_1,\ldots,x_{d+1}) = (x_1,\ldots,x_d)$. For a square matrix $A$ of size $(d+1)\times (d+1)$ 
define the following pair of projection functions. 
The first function $h_{1,d}$ takes as input a square matrix $A$ of size $(d+1)\times (d+1)$ and 
returns a square matrix of size $d\times d$ that is the upper-left block of $A$.
The second function $h_{2,d}$ takes as input a square matrix $A$ of size $(d+1)\times (d+1)$ 
and returns the vector of size $d$ given by  $[a_{1,d+1} \cdots a_{d,d+1}]^T$ (the last column of $A$ minus the last element).

Let $s$ denote the size function, its domain of objects will be overloaded and understood from the context.
For $x \in \Z$, $s(x)$ is the length of the encoding of $x$ in base $2$. For $x\in\Q$ 
with $x = \frac{p}{q}$ with $p$ and $q$ coprime, we have $s(x) = \max(s(p),s(q))$. For an affine function $f$
we define the size of $f(\vec x) = A\vec x+\vec b$ (where all entries of $A$ and $\vec b$ are rationals) as:
$s(f) = \max(\max_{i,j}(s(a_{i,j})),\max(s(b_i)))$. We define the size of a polyhedron $r$ defined
by $A\vec{x}\leqslant\vec{b}$ as: $s(r)=\max(s(A),s(\vec{b}))$.

We define the size of a piecewise affine function $f$ as:
$s(f) = \max_i(s(f_i),s(r_i))$ where $f_i$ denotes the restriction of $f$ to $r_i$ the $i^{th}$ region.

We define the \emph{signature} of a point $\vec x$ as the sequence of indices of the regions traversed 
by the iterates of $f$ on $\vec x$ (that is, the region trajectory).

\subsection{\texttt{REACH-REGION-TIME} is in NP}

In order to solve a reachability problem, we will formulate it with linear
algebra. However a crucial issue here is that of the size of the numbers, especially
when computing powers of matrices. Indeed, if taking the $n^{th}$ power of
$A$ yields a representation of exponential size, no matter how fast our algorithm
is, it will run on exponentially large instances and thus be slow.

First off, we show how to move to homogeneous coordinates so that $f$ becomes piecewise
linear instead of piecewise affine.

\begin{lemma}\label{lem:affine_to_linear}
Assume that $f(\vec{x})=A\vec{x}+\vec{b}$ with $A=(a_{i,j})_{1\leqslant i,j\leqslant d}$
and let $y=A'(\vec{x},1)^T$ where $A'$ is the block matrix
$\begin{pmatrix}A&\vec{b}\\0&1\end{pmatrix}$.
Then $f(x)=g_d(A'(\vec{x},1)^T)$.
\end{lemma}

\begin{remark}
Notice that this lemma extends nicely to the composition of affine functions:
if $f(\vec{x})=A\vec{x}+\vec{b}$ and $h(\vec{x})=C\vec{x}+\vec{d}$ then
$h(f(x))=g_d(C'A'(\vec{x},1)^T)$.
\end{remark}

We can now state the main lemma, namely that the size of the iterates of $f$
vary linearly in the number of iterates, assuming that $f$ is piecewise affine.

\begin{lemma}\label{lem:sz_paf_it}
Let $d\geqslant 2$ and $f\in PAF_d$. Assume that all the coefficients of $f$ on all
regions are rationals. Then for all $t\in\N$,
$s(f^{[t]})\leqslant (d+1)^2s(f)pt+(t-1)\lceil\log_2(d+1)\rceil$ where $p$ is the number
of regions of $f$. This inequality holds even if all rationals are taken to have the
same denominator.
\end{lemma}

\begin{proof}

Using \lemref{lem:affine_to_linear}, we get that $f^{[t]}(\vec{x})=g_d(h^{[t]}([\vec{x}\;1]^T))$,
where $h$ is a piecewise linear function in dimension $d+1$ such that $s(h)=s(f)$.
We show this result by induction on $t$ for $h$. The result then follows for $f$. In all cases we take all rationals to have the same denominator.

In the case $t=1$, it suffices to see that taking all rationals to have the same denominator involves 
multiplying the numerator and denominators by at most the lowest
common multiple of all numbers, and hence is at most $2^{s(f)(p(d+1)^2)}$.
Indeed the greatest number is $2^{s(h)}$ by definition, and there are $(d+1)^2$ numbers per region (the entries of the matrix).

Assume the result is true for $t\in\N$. Let $\vec{y}\in\Q^{d+1}$.
Then $h^{[t+1]}(\vec{y})=B_{t+1}\cdots B_1\vec{y}$, where $B_i$'s are the matrices corresponding to some regions of $h$. 
In particular, $s(B_i)\leqslant s(h)$. From the induction hypothesis we can assume that all rationals have the same denominator and we get that
$s(B_t\cdots B_1)\leqslant (d+1)^2s(h)pt+(t-1)\lceil\log_2(d+1)\rceil$.
It follows\footnote{Use elementary properties of the size function: $s(xy)\leqslant s(x)+s(y)$, $s(x_1+\cdots+x_k)\leqslant s(k)+\max_ks(x_k)$} that 
for any $1\leqslant i,j\leqslant d+1$: 

\begin{align*}
s((B_{t+1}\cdots B_1)_{i,j})&=s\left(\sum_{k=1}^{d+1}\left(B_{t+1}\right)_{i,k}\left(B_t\cdots B_1\right)_{k,j}\right)\\
&\leqslant\lceil\log_2(d+1)\rceil+s(B_{t+1})+s(B_t\cdots B_1)\\
&\leqslant\lceil\log_2(d+1)\rceil+s(h)+(d+1)^2s(h)pt+(t-1)\lceil\log_2(d+1)\rceil\\
&\leqslant (d+1)^2s(h)p(t+1)+t\lceil\log_2(d+1)\rceil
\end{align*}

This shows the result for the particular region where $y$ belongs. Since the bound does not
depend on $y$ and $h^{[t+1]}$ has finitely many regions, it is true for all regions of
$h^{[t+1]}$.
\qed\end{proof}

Finally, we need some result about the size of solutions to systems of linear inequalities.
Indeed, if we are going to quantify over the existence of a solution of polynomial size,
we must ensure that the size constraints do not change the satisfiability of the system.

\begin{lemma}[\cite{KComputingReals}]\label{lem:sys_poly_size_sol}
Let $A$ be a $N\times d$ integer matrix  and $\vec{b}$ an integer vector. If the
$A\vec{x}\leqslant\vec{b}$ system admits a solution, then there exists a rational solution
$x_s$ such that $s(x_s)\leqslant (d+1)L+(2d+1)\log_2(2d+1)$ where $L=\max(s(A),s(b))$.
\end{lemma}

\begin{proof}
See Theorem 5 of \cite{KComputingReals}: $s(x_s)\leqslant s\left((2d+1)!2^{L(2d+1)}\right)$.
\qed\end{proof}

Putting everything together, we obtain a fast nondeterministic
algorithm to solve \texttt{REACH-REGION-TIME}.
The nondeterminism allows us to choose a signature for the solution.
Once the signature is fixed, we can write it as a linear program of reasonable size using
\lemref{lem:sz_paf_it} and solve it.
The remaining issue is the one of the size of solution but fortunately
\lemref{lem:sys_poly_size_sol} ensures us that there is a small solution that can be found quickly.

\begin{theorem}\label{th:reach_in_np}
\texttt{REACH-REGION-TIME} is in NP.
\end{theorem}

\begin{proof}
The idea of the proof is to nondeterministically choose a signature for a solution,
that is a sequence of regions for the iterates of the solution. We then build a
system of linear inequalities stating that a point $\vec{x}$ belongs to the initial region
and that the iterates match the signature chosen and finally that the iterates reach the
final region. Using the results of the previous section, we can build this system
in polynomial time and solve it in non-deterministic polynomial time. Here is an
outline of the algorithm:
\begin{itemize}
\item Non-deterministically choose $t\leqslant T$
\item Non-deterministically choose regions $r_1,\ldots,r_{t-1}$ regions of $f$
\item Define $r_0=R_0$ the initial region and $r_t=R$ the final region
\item Build $(S)$ the system $A\vec{x}\leqslant\vec{b}$ stating that the signature of $\vec{x}$ matches $r$
\item Non-deterministically choose $\vec{x_s}$ a rational of polynomial size in the size of $(S)$
\item Accept if $A\vec{x_s}\leqslant\vec{b}$
\end{itemize}

We have two things to prove. First we need to show that this algorithm indeed has non-deterministic
polynomial running time. Second we need to show that it is correct. Recall that $T$ is a
unary input of the problem.

The complexity of the algorithm is clear, assuming that $(S)$ is of polynomial size.
Indeed verifying that a rational point satisfies a system of linear inequalities
with rationals coefficients can be done in polynomial time.

We build $(S)$ this way: $(S)=\cup_{i=1}^t(S_i)$ where $(S_i)$ states that $f^{[i]}(\vec{x})\in r_i$.
Since we choose a signature of $\vec{x}$ we know that if $\vec{x}$ satisfies the system then from
\lemref{lem:affine_to_linear}
$f^{[i]}(\vec{x})=g_d\left(A'_{i-1}\cdots A'_1(\vec{x},1)^T\right)$ where $A'_j$ is the
matrix corresponding to the region $r_j$. Write $C_i=A'_{i-1}\cdots A'_1$ and define $(S_i)$
by the system $g_d\left(C_i(\vec{x},1)^T\right)\in r_i$. Since $r_i$ is a polyhedron, $(S_i)$ is indeed
a system of linear inequalities\footnote{More precisely if $r_i$ is defined by $P_i(\vec{x},1)^T\leqslant0$
then $(S_i)$ is the system $P_iC_i(\vec{x},1)^T\leqslant0$}.

We can now see that $S$ is of polynomial size using \lemref{lem:sz_paf_it}.
Indeed, $s(C_i)\leqslant s(f^{[i]})\leqslant\text{poly}(s(f),i)$, thus
$s((S_i))\leqslant s(C_i)+s(r_i)\leqslant\text{poly}(s(f),i)$ because the
description of the regions is part of the size of $f$. And finally
$s((S))\leqslant\text{poly}(s(f),t)$.

The correctness follows from the construction of the system and \lemref{lem:sys_poly_size_sol}.
More precisely we show that $\vec{x}\in(S)$ if and only if $\forall i\in\{0,\ldots,t\}, f^{[i]}(\vec{x})\in r_i$.
Indeed, $(S)\Leftrightarrow\forall i\in\{0,\ldots,t\},\vec{x}\in(S_i)$ and
by definition $(S_i)\Leftrightarrow f^{[i]}(\vec{x})\in r_i$ since $g_d(C_i(\vec{x},1)^T)=f^{[i]}(\vec{x})$.
Then by \lemref{lem:sys_poly_size_sol}, we get that
$\exists \vec{x}\in(S)\Leftrightarrow\exists \vec{x}\in(S)\text{ and }s(\vec{x})\leqslant\text{poly}(s((S)))$.
\qed\end{proof}

\section{Other Bounded Time Results}

In this section, we give succinct proofs of the other result mentioned in the
introduction about \texttt{CONTROL-REGION-TIME}. The proof is based on the same
arguments as before.

\begin{theorem}\label{th:coreach_conp_hard}
Problem \texttt{CONTROL-REGION-TIME} is coNP-hard for $d\geqslant2$.
\end{theorem}

\begin{proof}
The proof is exactly the same except for two details:
\begin{itemize}
\item we modify $f$ over $R_{n+1}$ as follows: divide $R_{n+1}$ in three regions:
$R_{low}$ that is below $R_{fin}$, $R_{fin}$ and $R_{high}$ that is above $R_{fin}$.
Then build $f$ such that $f(R_{low})\subseteq R_{low}$, $f(R_{fin})\subseteq R_{fin}$
and $f(R_{high})\subseteq R_{low}$.
\item we choose a new final region $R'_{fin}=R_{low}$.
\end{itemize}
Let $\mathcal{I}=(B,A_1,\ldots,A_n)$ be an instance of NOSUBSET-SUM,
let $\mathcal{J}$ be the corresponding instance of CONTROL-REGION-TIME we just built.
We have to show that $\mathcal{I}$ has no subset sum if and only if $\mathcal{J}$
is ``controlled''. This is the same as showing that $\mathcal{I}$ has a subset sum
if and only if $\mathcal{J}$ has points never reaching $R'_{fin}$.
    
Now assume for a moment that the instance is in SUBSET-SUM (as opposed to NOSUBSET-SUM),
then by the same reasoning as the previous proof, there will be a point that reaches the
old $R_{fin}$ region (and is disjoint from $R'_{fin}$). And since
$R_{fin}$ is a $f$-stable region, this point will never reach $R'_{fin}$.

And conversely, if the control problem is not satisfied, necessarily there is a point
whose trajectory went
through the old $R_{fin}$ (otherwise if would have reached either $R_{low}=R'_{fin}$ or $R_{high}$
but $f(R_{high})\subseteq R_{low}$). Now we proceed as in the proof of \thref{th:reach_np_hard}
to conclude that there is a subset that sums to $B$, and thus $\mathcal{I}$ is satisfiable.
\qed\end{proof}

\begin{theorem}\label{th:coreach_conp}
Problem \texttt{CONTROL-REGION-TIME} is in coNP for $d\geqslant2$.
\end{theorem}

\begin{proof}
Again the proof is very similar to that of \thref{th:reach_in_np}: we have to build
a non-deterministic machine that accepts the ``no'' instances. The algorithm is
exactly the same except that we only choose signatures that avoid the final region
(as opposed to ending in the final region) and are of maximum length (that is $t=T$ as opposed
to $t\leqslant T$).
Indeed, if there is a such a trajectory, the problem is not satisfied. And for the same
reasons as \thref{th:reach_in_np}, it runs in non-deterministic polynomial time.
\qed\end{proof}

\section{Fixed Precision Results}

\begin{theorem}\label{th:reach_reg_prec_pspace_hard}
\texttt{REACH-REGION-PRECISION} and \texttt{CONTROL-REGION-PRECISION} are $PSPACE$-hard.
\end{theorem}

\begin{proof}
Consider a polynomial space Turing machine $\mathcal{M}=(Q,\Sigma,B,q_0,F,\delta)$.
Without loss of generality, we can assume that $F=\{q_f\}\subset Q$ (there is a single
accepting state) and that the working alphabet is
$\Sigma=\{0,1,2,\ldots,\beta\}$, assuming that $B=0$ is the blank
character, and $\delta:Q\times\Sigma\rightarrow Q\times\Sigma\times\{\triangleleft,\square,\triangleright\}$
is complete. We also assume the set of internal states to be such that $Q\subseteq\Sigma^m$ for some $m$.

Let $c$ be an instantaneous configuration (sometimes also called
\emph{ID} for \emph{Instataneous Description}) of $\mathcal{M}$.
We write $c=(x,\sigma,q,y)$ where $x$ (resp. $y$) is
encoding the part of the tape on the left (resp. right) of the
head, $\sigma$ is the symbol under the  head, and $q$ is the state of the
machine. Specifically, if the non-blank part of the tape is
$s_{-n} \dots s_{-1} \sigma s_1 s_2 \dots s_m$, with the head in front
of $\sigma$, then $x$ is encoded as word $s_{-1} s_{-2} \dots s_{-n}$,
and $y$ as word $s_1 s_2 \dots s_m$.

Define the encoding of configuration $c$ as $\langle c\rangle=(\langle x\rangle,\langle q\sigma y\rangle)$
where for any word $w\in\Sigma^*$, $\langle
w\rangle=\sum_{i=1}^{|w|} 2 w_i(2^\gamma)^{-i}$, where $\gamma$ is such that
$2 \beta+1 \leqslant 2^\gamma$. 
Define regions $R_{\alpha,q,\sigma}=[\langle\alpha\rangle]\times[\langle q\sigma\rangle]$
where $[\langle w\rangle]$ is a shortcut for $[\langle w\rangle]=[\langle w\rangle,\langle w\rangle+(2^\gamma)^{-|w|}]$.
Intuitively, $R_{\alpha,q,\sigma}$ contains all configurations in state $q$, with symbol $\sigma$ under the head
and symbol $\alpha$ immediately at the left of the head.
By construction $\langle c\rangle\in R_{x_1,q,\sigma}$ with the above notations.
Finally, for $\alpha,\sigma\in\Sigma$, $q\in Q$, define $f$ on region $R_{\alpha,q,\sigma}$ by:
\[
f(a,b)=\begin{cases}a2^{-\gamma}+\langle\sigma'\rangle,(b-\langle q\sigma\rangle)2^\gamma+\langle q'\rangle)
    &\text{if }\delta(q,\sigma)=(q',\sigma',\triangleright)\\
   (a,b-\langle q\sigma\rangle+\langle q'\sigma'\rangle)
    &\text{if }\delta(q,\sigma)=(q',\sigma',\square)\\
(a-\langle\alpha\rangle)2^\gamma,(b-\langle q\sigma\rangle)2^{-\gamma}+\langle q'\alpha\sigma'\rangle)
    &\text{if }\delta(q,\sigma)=(q',\sigma',\triangleleft)
\end{cases}
\]

It is clear from the definition that $f$ is piecewise affine over its domain of definition.
Let $T$ be the function corresponding to one step of computation of
$\mathcal{M}$: $T$ is acting on 
configurations and maps any configuration $c$ to the
corresponding next configuration according to the program of
$\mathcal{M}$. A simple case analysis shows that for any configuration
$c$, $\langle T(c)\rangle=f(\langle c\rangle)$, using the fact that the blank
character $B$ was chosen to be $0$.

Observe furthermore that by the choice of the encoding, all the
regions $R_{\alpha,q,\sigma}$ are closed and at positive distance from each other:
$R_{\alpha,q,\sigma} \cap R_{\alpha',q',\sigma'}
= \emptyset$ whenever $(\alpha,q,\sigma) \neq (\alpha',q',\sigma')$.
It follows that $f$ can
be easily extended to a continuous piecewise linear function defined
over the whole domain
$[0,1]$ (similar arguments are used in \cite{KCG94}). By construction it will still satisfy  that for any configuration
$c$, $\langle T(c)\rangle=f(\langle c\rangle)$.

We can now state the reduction from problem
\texttt{LINSPACE-WORD}: consider an instance $(\mathcal{M},w)$
of this decision problem.  Define
$\varepsilon=(2^\gamma)^{-({|w|+m+1)}}$ where the choice of $m$ was explained above.
Then for any configuration $c$
reachable from the initial configuration
$c_0=(\varepsilon,q_0,w_1,$ $w_2\cdots w_{|w|})$, we have the stronger
property that
$\langle c\rangle=\lfloor\frac{\langle
  c\rangle}{\varepsilon}\rfloor\varepsilon$.
Indeed, by assumption
the machine never uses more space than the size of the input, thus the left and
right part of tape are always smaller than $|w|$ at any point during the computation,
and we simply need an extra $m+1$ space to store the current state of the machine.
In other words, rounding to $\varepsilon$ does not perturbate the
simulation. Consequently, we get that for any reachable configuration
$c$, $\langle T(c)\rangle=f_\varepsilon(\langle c\rangle)$.

Define $R_0=\{\langle c_0\rangle\}$ and $R=[0,1]\times[\langle q_f\rangle]$ that are convex regions. Then
the instance $(f,R_0,R,\varepsilon)$ of \texttt{REACH-REGION-PRECISION} is satisfiable
if and only if $(\mathcal{M},w)$ belongs to problem
\texttt{LINSPACE-WORD}. One easily checks that $(f,R_0,R,\varepsilon)$ has polynomial
size in the size of $(\mathcal{M},w)$.

The same instance also works
for \texttt{CONTROL-REGION-PRECISION}. If we want to make $R_0$ a region with non-empty interior,
just take a ball of radius smaller that $\varepsilon(2^\gamma)^{-2}$ around $\langle c_0\rangle$
so that any input error is removed after the first application of the function $f_\varepsilon$.
\qed\end{proof}

\begin{theorem}\label{th:reach_reg_prec_in_pspace}
\texttt{REACH-REGION-PRECISION} and \texttt{CONTROL-REGION-PRECISION} are in $PSPACE$.
\end{theorem}

\begin{proof}
Let $N=\lfloor \varepsilon^{-1}\rfloor$ and consider the graph $G=(V,E)$ where
\[V=\{0,\ldots,N-1\}^d
\qquad C_\alpha=\prod_{k=1}^d\big[\alpha_k\varepsilon,(\alpha_k+1)\varepsilon\big[\quad(\alpha\in V)\]
\[E=\left\{(\alpha,\beta)\thinspace|\thinspace f(C_\alpha)\cap
  C_\beta\neq\varnothing\right\} \qquad  S=\{\alpha\thinspace|\thinspace f(R_0)\cap
C_\alpha\neq\varnothing\} \]
\[ T=\{\alpha\thinspace|\thinspace R\cap C_\alpha\neq\varnothing\}\]
We can now restate our reachability problem in the graph $G$ as an accessibility problem
from $S$ to $T$. This can be done in space logarithmic in the size of
the graph, using the fact that accessibility in a graph with $N$
vertices can be done in non-deterministic 
space $O(\log N)$, and using the fact that $NSPACE(N)=SPACE(N^2)$
(Savitch's Theorem) \cite[Theorem 8.5]{Sip97}. 
Since the graph is of size $\bigO{N^d}$, this requires space
$\bigO{\log^2 N}=\bigO{-\log^2 \varepsilon}=\bigO{n^2}$ if $\varepsilon=2^{-n}$.
Also note that computing the transitions of the graph is fast since $f$ is a piecewise
affine function. The same proof applies to \texttt{CONTROL-REGION-PRECISION}
except we now want to know if for every vertex in $S$ there is a path to $T$.
\qed\end{proof}



\bibliographystyle{splncs}
\bibliography{references,perso,bournez}

\end{document}